\documentclass[review]{elsarticle}
\usepackage[margin=3cm]{geometry} 

\usepackage[utf8]{inputenc}
\usepackage{tudacolors}
\graphicspath{{figures/}} 
\usepackage{caption}
\usepackage{subcaption}
\usepackage{amsmath,amsthm,amssymb}
\usepackage{hyperref} 
\usepackage{siunitx}

\usepackage{pxfonts}
\usepackage[T1]{fontenc}
\usepackage{lmodern}

\newtheorem{lemma}{Lemma}
\newtheorem{corollary}{Corollary}

\newcommand{\trans}{^{\!\top}}

\newcommand{\ddt}{\frac{\mathrm{d}}{\mathrm{d}t}}
\newcommand{\tr}{\mathrm{tr}}
\newcommand{\dx}{\mathrm{d}\mathbf{x}}
\newcommand{\dy}{\mathrm{d}\mathbf{y}}

\newcommand{\R}{\ensuremath{\mathbb{R}}}
\newcommand{\Hcurl}{\ensuremath{H\left(\textrm{curl};\Omega\right)}}
\newcommand{\Hzerocurlt}{\ensuremath{H_0\left(\textrm{curl};\Omega[t]\right)}}
\newcommand{\Hone}{\ensuremath{H^1\left(\Omega\right)}}
\newcommand{\Hzeroonet}{\ensuremath{H_0^1\left(\Omega[t]\right)}}

\renewcommand{\vec}[1]{\ensuremath{\mathbf{#1}}}
\newcommand{\curl}[1]{\ensuremath{\nabla \times #1}}
\renewcommand{\div}[1]{\ensuremath{\nabla \cdot #1}}
\newcommand{\grad}[1]{\ensuremath{\nabla #1}}
\newcommand{\fembasis}{\ensuremath{N}}

\newcommand{\stiff}{\ensuremath{\mathbf{K}}}
\newcommand{\mass}{\ensuremath{\mathbf{M}}}
\newcommand{\stiffgrad}{\ensuremath{\overline{\mathbf{K}}}}
\newcommand{\massgrad}{\ensuremath{\overline{\mathbf{M}}}}
\newcommand{\stiffHone}{\stiffgrad}
\newcommand{\massHone}{\massgrad}

\newcommand{\parat}{\ensuremath{t}}

\newcommand{\map}{\ensuremath{\mathbf{F}}}
\newcommand{\maptil}{\ensuremath{{\mathbf{G}}}}
\newcommand{\maptilF}{\ensuremath{\tilde{\mathbf{F}}}}
\newcommand{\onref}[1]{\ensuremath{\hat{#1}}}
\newcommand{\ndim}{\ensuremath{3}}
\newcommand{\dimSpline}{\ensuremath{n_\mathrm{dim}}}
\newcommand{\Ngeo}{\ensuremath{n_\mathrm{geo}}}
\newcommand{\Ndof}{{\ensuremath{n_{\mathrm{dof}}}}}

\usepackage{glossaries}

\newacronym{fem}{FEM}{finite element method}
\newacronym{iga}{IGA}{isogeometric analysis}
\newacronym{cad}{CAD}{computer-aided design}

\begin{document}

\title{On the computation of analytic sensitivities of eigenpairs in isogeometric analysis}

\author{Anna Ziegler$^{1,2}$, 
    Melina Merkel$^{1,2}$, 
    Peter Gangl$^{3}$, and 
    Sebastian Schöps$^{1,2}$\\[0.5em]
    \footnotesize$^{1}$ Computational Electromagnetics Group, Technische Universität Darmstadt, 64289 Darmstadt, Germany\\
	\footnotesize$^{2}$ Centre for Computational Engineering, Technische Universität Darmstadt, 64293 Darmstadt, Germany\\
	\footnotesize$^{3}$ Johann Radon Institute for Computational and Applied Mathematics, Austrian Academy of Sciences, 4040 Linz, Austria
}

\date{\today}

\begin{abstract}
The eigenmodes of resonating structures, e.g., electromagnetic cavities, are sensitive to deformations of their shape.
In order to compute the sensitivities of the eigenpair with respect to a scalar parameter, we state the Laplacian and Maxwellian eigenvalue problems and discretize the models using isogeometric analysis.
Since we require the derivatives of the system matrices, we differentiate the system matrices for each setting considering the appropriate function spaces for geometry and solution.
This approach allows for a straightforward computation of arbitrary higher order sensitivities in a closed-form. 
In our work, we demonstrate the application in a setting of small geometric deformations, e.g., for the investigation of manufacturing uncertainties of electromagnetic cavities, as well as in an eigenvalue tracking along a shape morphing. 
\end{abstract}
\maketitle

\section{Introduction}
In 2005, Hughes et al. \cite{Hughes_2005aa} introduced the concept of \gls{iga} for the discretization of partial differential equations (PDEs). It is a particular version of the \gls{fem} using the same spline-based basis functions that are employed in \gls{cad} systems, capable of representing many engineering geometries exactly. Another advantage of \gls{iga} is the inter-element smoothness. As a consequence, IGA promises faster convergence, i.e., fewer degrees of freedom are needed for the same accuracy, when compared to traditional \gls{fem} approaches \cite{Hughes_2014aa, Cottrell_2009aa}. The method has been extended beyond applications in solid mechanics for example to electromagnetism~\cite{Vazquez_2010aa} and fluid mechanics \cite{Hoang_2017aa}. In such applications, the (generalized) eigenvalue problem and the spectral properties of the discretization have been a particular topic of interest in the scientific community, see for example \cite{Cottrell_2006aa}. 

In the present paper, we investigate the solution of the eigenvalue problem on a parametrized domain $\Omega[t]$ and its derivative with respect to the parameter~$t$. This is for example relevant when optimizing resonant structures \cite{Lewis_2003aa}, in uncertainty quantification of eigenvalue problems \cite{Georg_2019aa}, or for mode tracking~\cite{Jorkowski_2018aa}. If considering geometries that are exactly represented by our spline-based CAD basis functions, then the geometric deformation and its derivatives can be explicitly expressed in terms of parameter-dependent weights and control points. In the following, we compute those derivatives of the stiffness and mass matrices in $H^1(\Omega)$ and $\Hcurl$ that are necessary to eventually determine the derivatives of eigenvalues and -vectors, {allowing their efficient and accurate approximation also for similar geometries by means of a Taylor expansion.}

The approach we follow for obtaining derivatives of stiffness and mass matrices with respect to geometric changes is closely related to the computation of shape derivatives of PDE-constrained optimization problems \cite{DZ2, SZ}. Also in that context, domains are subject to a transformation (usually represented by the action of some vector field) and sensitivities are obtained by first transforming back to the original domain and differentiating the arising integrand. In the context of IGA, the deformation vector field is typically defined in terms of the given control points and shape deformations are therefore given by a change in control points. In~\cite{Qian_2010aa}, shape sensitivities in an IGA context are considered on a discrete level where the differentiation is carried out with respect to both control points and weights in a NURBS description. Later it was shown in \cite{Fuseder_2015ab} that, under certain assumptions, discretization and shape differentiating commute in the context of IGA, i.e., shape derivatives actually amount to differentiation with respect to control points. For an application of IGA-based shape optimization in the context of electromagnetics, we refer the interested reader to~\cite{MerkelGanglSchoeps2021}.

The paper is structured as follows. 
We introduce the Laplacian and Maxwellian eigenvalue problems and derive their weak formulations and discrete eigenvalue problems in the following.
In Section~\ref{sec:iga}, we review the basics of isogeometric analysis and state the suitable function spaces for both problem types.
Subsequently, we present the closed-form formulation of the derivatives of the system matrices on transformed domains in Section~\ref{sec:diffiga}.
We demonstrate the formulation of the shape morphing as a suitable transformation and extend this by determining higher order derivatives. 
In Section~\ref{sec:numerics}, we show their application in numerical examples.
Finally, we conclude our work in Section~\ref{sec:conclusion}.
\subsection{Problem Setting}
Given a bounded and simply connected domain $\Omega[t] \in \R^\ndim$ parametrized by $t\in[0,1]$, with Lipschitz continuous boundary $\partial\Omega[t]$, we consider the Laplacian and Maxwellian eigenvalue problems with Dirichlet boundary conditions, see e.g. \cite{Boffi_2010aa}. In the first case, we look for eigenpairs $u_t\neq0$ and $\lambda_t$ such that
\begin{equation}\label{eq:Laplace}
	\begin{aligned}
		{-}\div{\left(\grad{u_t}\right)} &= \lambda_t^2 u_t && \text{in }\Omega[t]\\
		u_t &= 0 && \text{on }\partial\Omega[t]
	\end{aligned}
\end{equation}
and in electromagnetics, we seek electric field strengths $\vec{E}_t\neq0$ and eigenvalues~$\lambda_t$ such that
\begin{equation}\label{eq:Maxwell}
	\begin{aligned}
		\curl{\left(\curl{\vec{E}_t}\right)} &= \lambda_t^2 \vec{E}_t && \text{in }\Omega[t]\\
		\vec{E}_t\times\vec{n} &= 0 && \text{on }\partial\Omega[t]
	\end{aligned}
\end{equation}
with outward pointing normal vector $\mathbf{n}$. The resonant frequency of the field in vacuum is easily deduced by $f_t=\frac{\sqrt{\lambda_t}c_0}{2\pi}$ where $c_0$ is the speed of light. 
Let us number the eigenpairs, i.e., $(u_{t,m}$, $\lambda_{t,m})$ and $(\vec{E}_{t,m}$, $\lambda_{t,m})$, such that $\lambda_{t,m-1}\leq \lambda_{t,m}$ for $m=2,3,\ldots$ and note that each pair inherits the dependency on the parameter $t$ from the computational domain~$\Omega[t]$. 

\subsection{Weak formulation}
The well-known variational formulations of the two eigenvalue problems are, see e.g. \cite{Boffi_2010aa}: 
find $\lambda_t\in \R$ and $u_t \in \Hzeroonet$ such that
\begin{equation}\label{eq:Laplace-eig-var}
 \left(\grad{u_t},\grad{\fembasis_t}\right) 
  = \lambda_t \left(u_t,\fembasis_t\right) \quad \forall \fembasis_t \in \Hzeroonet, 
\end{equation}
or find $\lambda_t\in \R$ and $\vec{E}_t \in \Hzerocurlt$ such that
\begin{equation}\label{eq:Maxwell-eig-var}
 \left(\curl{\vec{E}_t},\curl{\vec{\fembasis}_t}\right) 
  = \lambda_t \left(\vec{E}_t,\vec{\fembasis}_t\right) \quad \forall \vec{\fembasis}_t \in \Hzerocurlt, 
\end{equation}
where we made use of the usual function spaces $\Hzeroonet$ and $\Hzerocurlt$ of square-integrable fields with vanishing trace and tangential trace at the boundary, respectively. For further information on function spaces in the
context of Maxwell's equations, the reader is referred to \cite{Monk_2003aa}.

Let us follow the Ritz-Galerkin procedure and introduce a sequence of finite-dimensional spaces  $W[t] \subset \Hzeroonet$ and $\mathbf{W}[t] \subset \Hzerocurlt$ to yield approximate solutions of dimension
$\Ndof = \mathrm{dim} \left(W[t]\right)$ or $\Ndof=\mathrm{dim} \left(\mathbf{W}[t]\right)$, respectively. 
Using basis functions $\fembasis_{t,j}\in W[t]$ and $\mathbf{\fembasis}_{t,j}\in\mathbf{W}[t]$, the unknown fields are approximated by linear combinations of basis functions 
\begin{equation}
	u_t \approx \sum_{j=1}^{\Ndof} u_{t,j} \fembasis_{t,j}
\quad\text{or}\quad
	\vec{E}_t \approx \sum_{j=1}^{\Ndof} u_{t,j} \mathbf{\fembasis}_{t,j}.
\end{equation}
In both cases, the discrete solution
$\mathbf{u}_t=[u_{t,1},\ldots,u_{t,\Ndof}]\trans$
is obtained by solving the generalized eigenvalue problems
\begin{equation}\label{eq:galerkin-substitution}
\stiffgrad[t] \vec{u}_t = \lambda_t^2 \massgrad[t] \vec{u}_t \qquad \mbox{ and } \qquad \stiff[t] \vec{u}_t = \lambda_t^2 \mass[t] \vec{u}_t
\end{equation}
where the stiffness and mass matrices are given by
\begin{align} \label{eq_stiffH1}
	\stiffgrad[t]_{i,j} &= \left(\grad{\fembasis_{t,i}}, \grad{\fembasis_{t,j}}\right), &
	\massgrad[t]_{i,j}  &= \left(\fembasis_{t,i}, \fembasis_{t,j}\right)
\intertext{ or in the case of electromagnetism}
	\stiff[t]_{i,j} &= \left(\curl{\mathbf{\fembasis}_{t,i}}, \label{eq_stiffHcurl} \curl{\mathbf{\fembasis}_{t,j}}\right), &
	\mass[t]_{i,j}  &= \left(\mathbf{\fembasis}_{t,i}, \mathbf{\fembasis}_{t,j}\right)
\end{align}
where $i,j\in\{1,\ldots,\Ndof\}$. In the following, we use only the notation $\stiff[t]$, $\mass[t]$ for the electromagnetic problem, however, the same steps also hold for the $H^1$ matrices $\stiffgrad[t]$, $\massgrad[t]$. Let us normalize the eigenvectors using some vector $\mathbf{u}_{\star}$ 
such that we can formulate the eigenvalue problem as the root finding problem
\begin{equation}
\begin{bmatrix}
\stiff[t] \mathbf{u}_t - \lambda_t \mass[t] \mathbf{u}_t \\
\mathbf{u}_{\star}\trans \mass[t] \mathbf{u}_t - 1
\end{bmatrix} = \begin{bmatrix}
0 \\	 0
\end{bmatrix} \label{eq:eigprob}
\end{equation}
with solutions $(\vec{u}_{t,m}, \lambda_{t,m})$ sorted such that $\lambda_{t,m-1}<\lambda_{t,m}$ for $m=2,\ldots,\Ndof$.
Note, that there are different choices to construct the spaces $W[t]$ and $\mathbf{W}[t]$. 
For \gls{fem} in the case of Maxwell equations, we refer the interested reader
to \cite{Boffi_2010aa,Monk_2003aa}. After discussing the derivatives {of the eigenvalues and eigenvectors}, we will
follow \cite{Cottrell_2006aa,Buffa_2010aa} and introduce the \gls{iga} discretization using B-spline spaces.

\subsection{Derivatives of eigenvalues and eigenvectors}
The $n$-th derivative of the $m$-th eigenvector and eigenvalue with respect to the parameter $t$, i.e.,
$$
	\mathbf{u}_{t,m}^{(n)}=\frac{\mathrm{d}^n}{\mathrm{d}t^n}\mathbf{u}_{t,m}
	\quad\text{and}\quad 
	\lambda_{t,m}^{(n)}=\frac{\mathrm{d}^n}{\mathrm{d}t^n}\lambda_{t,m}\;,
$$
are obtained by differentiating \eqref{eq:eigprob}, where we implicitly assume that these derivatives exist.
Let us assume $\mathrm{d}\mathbf{u}_{\star}/\mathrm{d}t=\textbf{0}$, then we obtain the
{first order derivative by solving the linear system
\begin{equation}
	\begin{bmatrix}
		\stiff[t] - \lambda_{t,m} \mass[t]  &\! -\mass[t]\mathbf{u}_{t,m} \\
		\mathbf{u}_{\star}\trans \mass[t] &0
	\end{bmatrix}
	\begin{bmatrix}
		\mathbf{u}_{t,m}'
		\\
		\lambda_{t,m}'
	\end{bmatrix}=
	\begin{bmatrix}
	-\stiff[t]'\mathbf{u}_{t,m}+\lambda_{t,m}\mass[t]'\mathbf{u}_{t,m} \\
	-\mathbf{u}_{\star}\trans \mass[t]'\mathbf{u}_{t,m}
	\end{bmatrix}
	\label{eq:eigprobderiv1}
\end{equation}
where $(\cdot)'$ denotes the first order derivative of an expression with respect to the parameter $t$.
Repeating the procedure, following~\cite{Nelson_1976aa}, we obtain the }
$n$-th {order} derivative by solving the linear system
\begin{equation}
	\begin{bmatrix}
		\displaystyle
		\vphantom{\binom{n}{k}}
		\stiff[t]-\lambda_{t,m} \mass[t] & -\mass[t] \mathbf{u}_{t,m}
		\\
		\displaystyle
		\vphantom{\binom{n}{k}}
		\mathbf{u}_{\star}\trans \mass[t]& 0
	\end{bmatrix}
	\begin{bmatrix}
		\displaystyle
		\mathbf{u}_{t,m}^{(n)}
		\vphantom{\binom{n}{k}}
		\\
		\displaystyle
		\lambda_{t,m}^{(n)}
		\vphantom{\binom{n}{k}}
	\end{bmatrix}
	= 
	 \begin{bmatrix}
		\displaystyle
		\mathbf{r}_{t,m}^{n}\\
		\displaystyle
		- \sum_{k=0}^{{n-1}} \binom{n}{k} \mathbf{u}_\star\trans\mass[t]^{(n-k)} \mathbf{u}_{t,m}^{(k)}
	\end{bmatrix} 
	\label{eq:eigprobderiv}
\end{equation}
with the right-hand-side
\begin{equation}
\begin{aligned}
\mathbf{r}_{t,m}^{n} := 
&
- \sum_{k=1}^{{n-1}} \binom{n}{k} \left(\stiff[t]^{(k)} 
- \sum_{j=0}^{k} \binom{k}{j} \lambda_{t,j}^{(j)} \mass[t]^{(k-j)}\right)\mathbf{u}_{t,j}^{(n-k)}  - \stiff[t]^{(n)}\mathbf{u}_{t,j}^{(0)}\\
&
+ \sum_{j=0}^{n-1} \binom{n}{j} \lambda_{t,j}^{(j)} \mass[t]^{(n-j)}\mathbf{u}_{t,j}^{(0)}\;.
\end{aligned}
\label{eq:rhs}
\end{equation}
In the general case, i.e., to obtain a derivative of degree $n_{\max}$, the system~\eqref{eq:eigprobderiv} must be solved repeatedly for~$n = 1, \ldots, n_{\max}$ since \eqref{eq:rhs} includes the derivatives of $\mathbf{u}_{t,m}$ and $\lambda_{t,m}$ up to degree~$n-1$. However, as the system matrix on the left-hand side remains the same for each order~$n$, the implementation can still be made efficient, e.g. when reusing the matrix factorization or preconditioner. 

To compute the corresponding derivatives of the eigenpair, we solve system~\eqref{eq:eigprobderiv}.
However, since the matrix $(\stiff[t] - \lambda_{t,m} \mass[t])$ does not have full rank in case of a multiplicity of an eigenvalue, the system needs special treatment~\cite{Dailey_1989aa}. 
In~\cite{Jorkowski_2020aa}, an algorithm is introduced which treats this rank defect and computes higher derivatives of multiple eigenvalues.

It can be noted from \eqref{eq:eigprobderiv} that the computation of (higher order) derivatives of the eigenvalues and eigenvectors involves (higher order) derivatives of the stiffness and mass matrices $\stiff[t]$ and $\mass[t]$. This will be discussed in Section \ref{sec:diffiga}.
\section{Isogeometric Analysis} \label{sec:iga}
Most CAD tools store the computational geometry by its boundary representation (b-rep). 
The b-rep is then internally parametrized by the union or intersection of several (possibly trimmed) NURBS patches \cite{Piegl_1987aa,Cohen_2001aa}. 
For our following isogeometric finite element analysis, we assume that the geometry is given by a volumetric representation consisting of untrimmed NURBS patches \cite{Massarwi_2019aa}. For example, such a description can be obtained by trivariate CAD kernels like IRIT \cite{elber_2021aa}. Each patch maps from a reference domain $\hat{\Omega}=[0,1]^{\ndim}$ into the (three-dimensional) physical domain $\Omega[t]$. 

We start from the {\dimSpline-dimensional} basis $\{ \hat{B}_{i,p} \}_{i=1}^{\dimSpline}$ of a one-dimensional B-spline space $\mathbb{S}^{p}_{\alpha}$ of degree $p$ and regularity~$\alpha$. The basis is constructed from a knot vector $\boldsymbol{\Xi} = (\xi_1, \xi_2, \dots, \xi_{\dimSpline+p+1})$ with \mbox{$0\leq\xi_1 \leq \xi_2 \leq \dots \leq \xi_{\dimSpline+p+1}\leq1$} using de Boor's algorithm~\cite{de-Boor_2001aa}
\begin{align}
    \hat{B}_{i,0}(\xi) &= \begin{cases}
        1 \quad \mathrm{if} \quad \xi_i \leq \xi < \xi_{i+1}\\
        0 \quad \mathrm{otherwise}
    \end{cases}
\intertext{and for $p>0$}
    \hat{B}_{i,p}(\xi) &= \frac{\xi - \xi_i}{\xi_{i+p} - \xi_i} \hat{B}_{i,p-1}(\xi) + \frac{\xi_{i+p+1} - \xi}{\xi_{i+p+1} - \xi_{i+1}} \hat{B}_{i+1,p-1}(\xi).
\end{align}
Please note the hat over $\hat{B}_{i,p}$. It symbolizes here and in the rest of the paper quantities and functions that are related to the reference domain. 

The construction can be straightforwardly generalized to the tensor product space
\begin{align}
    \mathbb{S}^{p_1,p_2,p_3}_{\alpha_1,\alpha_2,\alpha_3}(\hat{\Omega})
    :=
    S^{p_1}_{\alpha_1}(\hat{\Omega})
    \otimes
    S^{p_2}_{\alpha_2}(\hat{\Omega})
    \otimes
    S^{p_2}_{\alpha_3}(\hat{\Omega}).
\end{align}
Also, from the B-spline basis functions the NURBS curve  
\begin{align}
    \label{eq:NURBS}
\mathbf{S}[t](\xi) = \frac
	{\sum_{i=1}^{\dimSpline} {\hat{B}_{i,p}(\xi)w_i[t] \mathbf{P}_i[t]}}
	{\sum_{i=1}^{\dimSpline} {\hat{B}_{i,p}(\xi)w_i[t]}}
\end{align}
is constructed where both the control points $\mathbf{P}_i[t]$ and $w_i[t]$ may depend on a parameter $t$. 
Again, thanks to the tensor product construction, each volumetric patch is eventually given by a NURBS mapping from the reference space $\hat{\Omega}=[0,1]^3$ to the three-dimensional physical space. Possibly gluing several patches together, we have the (multipatch) mapping 
$$
\map[t]:\hat{\Omega}\to\Omega[t]
$$
for which we assume that it is (piecewise) smoothly invertible. Note, that this abstract parametrization is convenient for shape deformations since the change of the control points $\textbf{P}$ in terms of~$t$ facilitates a (smooth) change in the shape of the computational geometry, in particular for small deformations.
Let us denote the Jacobian of the transformation $\map$ by $\partial \map$ with 
    \begin{align}
        \partial \map_{i,j} = \frac{\partial \map_i}{\partial x_j} 
        \label{eq:jac}
	\end{align}
for $i, j = 1, \dots, \ndim$.
If we consider large deformations, the mappings must fulfill regularity assumptions, e.g., no intersections. We formalize this by requiring that the mapping is valid, i.e., $\det\left(\partial\map\right) > 0$.

We follow \cite{Buffa_2010aa,Buffa_2019ac} to define the compatible discretization spaces for the Laplace and Maxwell eigenvalue problems on the reference domain as
\begin{align}
    \hat{W}(\hat \Omega) 
    &:= \mathbb{S}^{p_1,p_2,p_3}_{\alpha_1,\alpha_2,\alpha_3}(\hat{\Omega})  
    \label{eq_def_What}
    \\
    \hat{\textbf{W}}(\hat{\Omega})
    &:=
    \mathbb{S}^{p_1-1,p_2,p_3}_{\alpha_1-1,\alpha_2,\alpha_3}(\hat{\Omega})
    \times
    \mathbb{S}^{p_1,p_2-1,p_3}_{\alpha_1,\alpha_2-1,\alpha_3}(\hat{\Omega})
    \times
    \mathbb{S}^{p_1,p_2,p_3-1}_{\alpha_1,\alpha_2,\alpha_3-1}(\hat{\Omega})
\intertext{and on a single patch in the physical domain the function spaces are given by}
    W[t](\Omega)
    &:=
    \bigl\{ \fembasis: \fembasis \circ\mathbf{F}[t] =
     \hat{\fembasis},
     \hat{\fembasis}\in\hat{W}(\hat{\Omega}) \label{eq_def_Wt}
    \bigr\}
    \\
    \textbf{W}[t](\Omega) 
    &:=
    \bigl\{
         \mathbf{\fembasis}
     :
     \mathbf{\fembasis}\circ\mathbf{F}[t]
     =
     (\partial \mathbf{F})^{-\top}\hat{\mathbf{\fembasis}},
     \hat{\mathbf{\fembasis}}\in\hat{\textbf{W}}(\hat{\Omega})
    \bigr\}. \label{eq_def_Wtvec}
\end{align}
Let 
$\lbrace\hat{\mathbf{\fembasis}}_j\rbrace_{j=1}^{\Ndof}$ be a (finite) basis for
$\hat{\mathbf{W}}$, then the set
$\lbrace\mathbf{\fembasis}_{j}\rbrace_{j=1}^{\Ndof}$ with $\mathbf{\fembasis}_j = (\partial  \mathbf{F}[t])^{-\top} \hat{\mathbf{\fembasis}}_j \circ \mathbf F[t]^{-1}$ is the corresponding basis for $\mathbf{W}[t]$ which can be used to discretize~\eqref{eq:Maxwell}. {Analogously, a basis $\lbrace \fembasis_j \rbrace_{j=1}^{\Ndof}$ of $W[t]$ is constructed from a basis $\lbrace \hat{\fembasis}_j \rbrace_{j=1}^{\Ndof}$ of $\hat W$ by $\fembasis_j = \hat{\fembasis_j}\circ \mathbf{F}[t]^{-1}$.}
{In the multipatch case, the spaces are glued, where in the Maxwellian case, we only ensure tangential continuity, see \cite{Buffa_2015aa}.} 

\begin{figure}
\centering
    \includegraphics[height=.09\textheight]{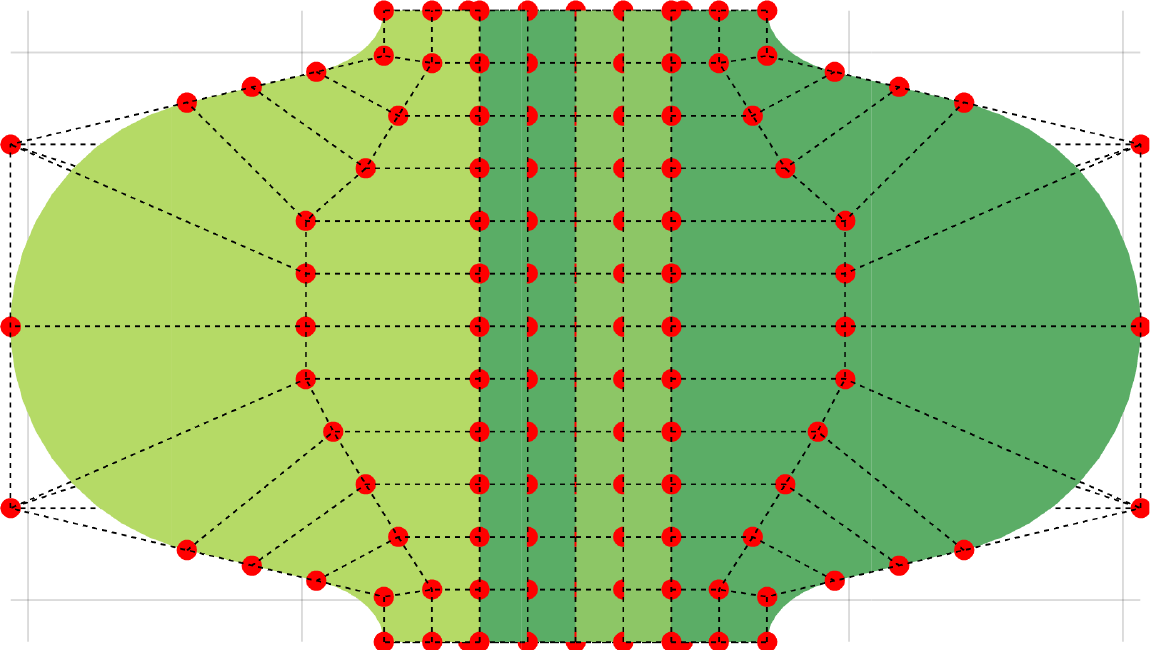}\hspace{1em}
    \includegraphics[height=.09\textheight]{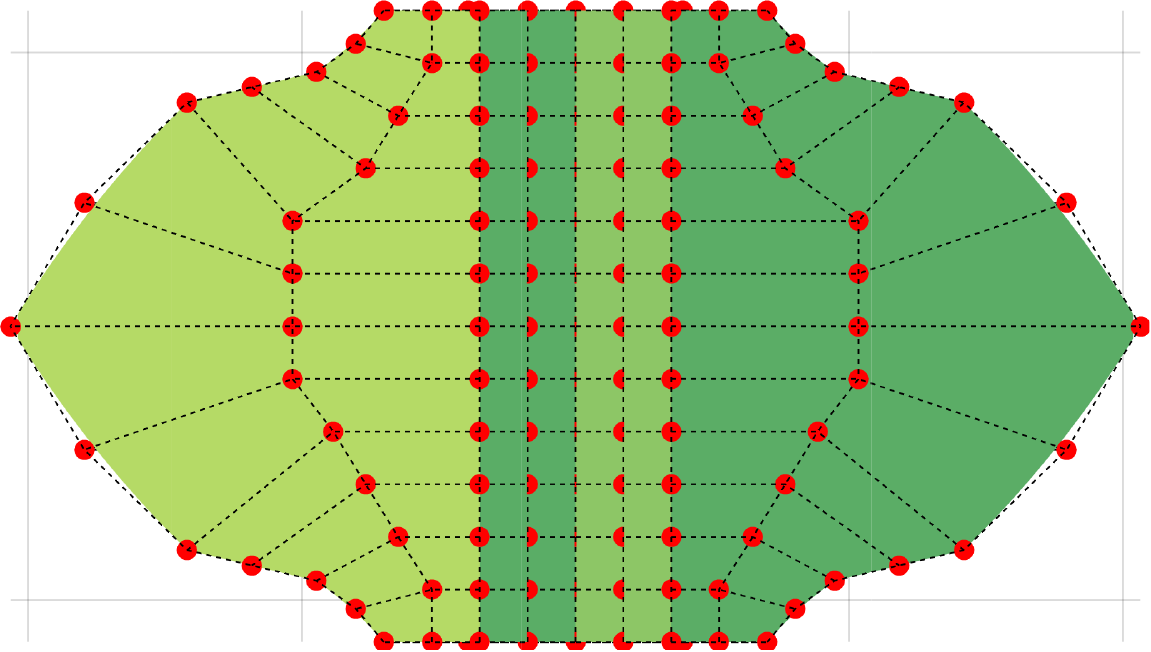}\hspace{1em}
    \includegraphics[trim = 3em 0em 3em 0em, clip, height=.09\textheight]{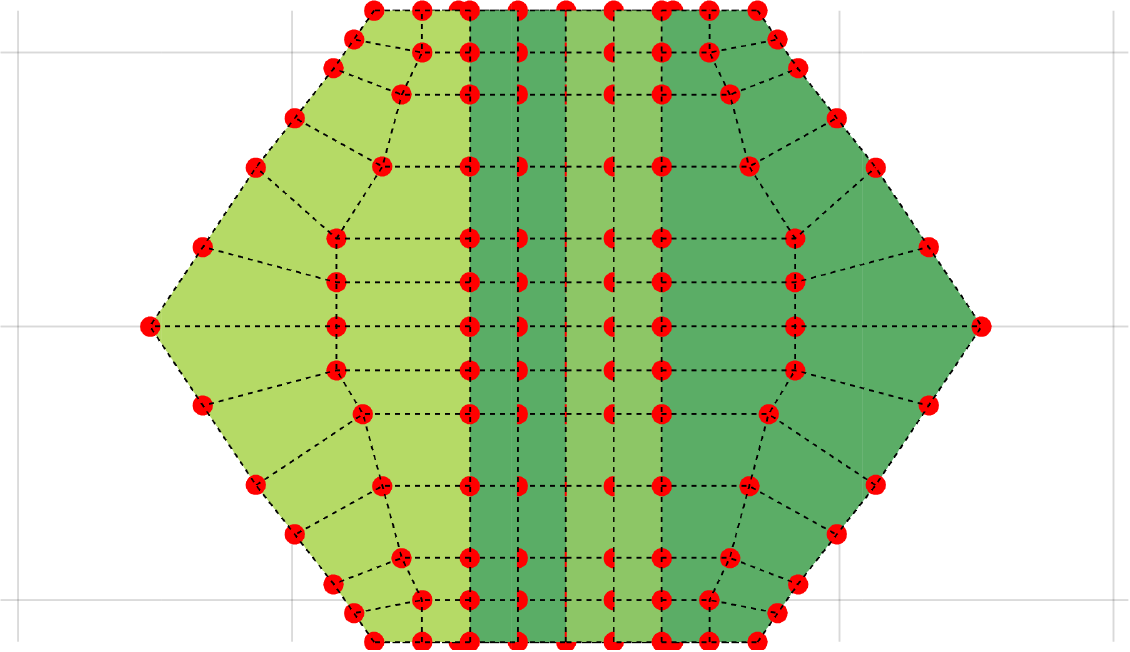}\hspace{1.5em}
    \includegraphics[trim = 8em 0em 8em 0em, clip, height=.09\textheight]{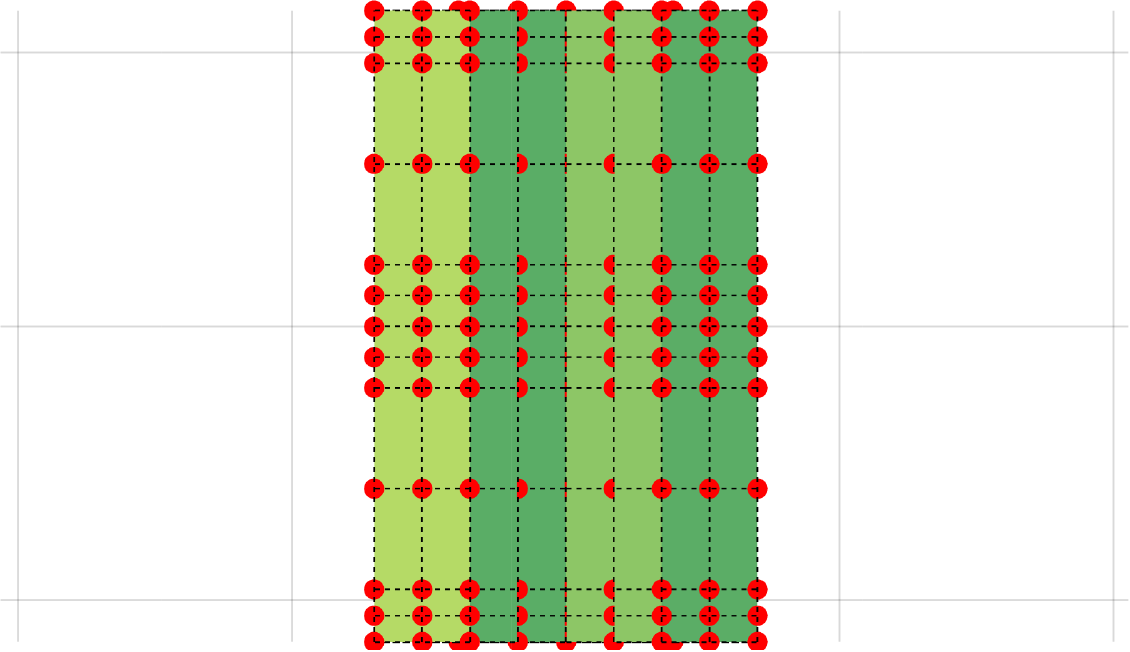}
    \caption{Illustration of a straightforward IGA mesh deformation which transforms the shape of a TESLA cell (left) to a pillbox shape (right). The full cavity cell is composed of 5 patches, see \cite[Section 3.6]{Corno_2017ad} for details on the multipatch approach for TESLA and pillbox cavities. Here, only two patches are shown for visualization purposes. A suitable mesh quality can be observed for all steps during the deformation.}
    \label{fig:IGAmeshDeformation}
\end{figure}

We note that considering large shape deformations with adequate mesh quality is relatively straightforward with IGA. 
{We demonstrate this by morphing a radiofrequency cavity into another shape. 
For this purpose, we consider the superconducting TESLA cavity~\cite{Aune_2000aa}, which is used for particle acceleration and is designed from elliptical shapes.
Its IGA control mesh with the control points marked in red, is illustrated in the cross section view in Fig.~\ref{fig:IGAmeshDeformation} on the left.
We investigate its deformation by convex combination of control points to the cylindrical so-called pillbox cavity, cf. Fig.~\ref{fig:IGAmeshDeformation} on the right, and observe a suitable mesh quality along all steps during the deformation.}

\section{Sensitivities of IGA matrices on parameter-dependent domains}\label{sec:diffiga}
    In this section, given a reference domain $\onref{\Omega}$, two fixed physical domains $\Omega_0$, $\Omega_1$ and a parameter-dependent physical domain $\Omega[\parat]$ continuously depending on a scalar parameter~$\parat$ satisfying $\Omega[0] = \Omega_0$ and $\Omega[1] = \Omega_1$, we compute the sensitivities of the matrices $\stiffHone[t]$ and $\massHone[t]$ defined in \eqref{eq_stiffH1} for the case of $H^1$ and $\stiff[t]$ and $\mass[t]$ defined in \eqref{eq_stiffHcurl} for the case of $\ensuremath{H\left(\mathrm{curl}\right)}$ with respect to the geometry parameter~$\parat$. Let $\map_0$ denote the mapping defining $\Omega_0$ from $\onref{\Omega}$, i.e., $\Omega_0 = \map_0(\onref{\Omega})$, and $\maptilF[t]$ the mapping defining $\Omega[t]$ from $\Omega_0$, i.e., $\Omega[t] = \maptilF[t](\Omega_0)$. By composition, it holds $\Omega[\parat] = \map[t](\onref \Omega)$ with $\map[t]:= \maptilF[t] \circ \map_0$, see also Figure \ref{fig_transformations}. 
	We will make use of the following well-known transformation rules \cite{Monk_2003aa}:	
	\begin{lemma}\label{lem_trafos} Let $\Omega$ a smooth domain, let $\maptil$ a smooth transformation and define $\Omega' := \maptil(\Omega)$. 
	       \begin{enumerate}
	           \item[a)]
	           Let $f \in H^1(\Omega')$. It holds
	           \begin{align} \label{eq_trafoGrad}
	               (\nabla f) \circ \maptil = (\partial \maptil)^{-\top} \nabla (f \circ \maptil).
	           \end{align}
	           \item[b)] Let $\vec{f} \in \ensuremath{H\left(\mathrm{curl};\Omega'\right)}$. It holds
	           \begin{align} \label{eq_trafoCurl}
                    (\mathrm{curl} \, \vec{f}) \circ \maptil = \frac{1}{\det(\partial \maptil)} \, \partial \maptil \, \mathrm{curl} \,(\partial \maptil^\top \vec{f}\circ \maptil).
            	\end{align}
	       \end{enumerate}
	\end{lemma}
	From Lemma \ref{lem_trafos} it follows that  $\vec{f} \in \ensuremath{H\left(\mathrm{curl};\Omega[t]\right)}$ if and only if $\partial \map[t]^\top \vec{f}\circ \map[t] \in \ensuremath{H\left(\mathrm{curl};{\hat \Omega} \right)}$. 
    
     \begin{figure}[htb]
        \centering
        \includegraphics[scale=1]{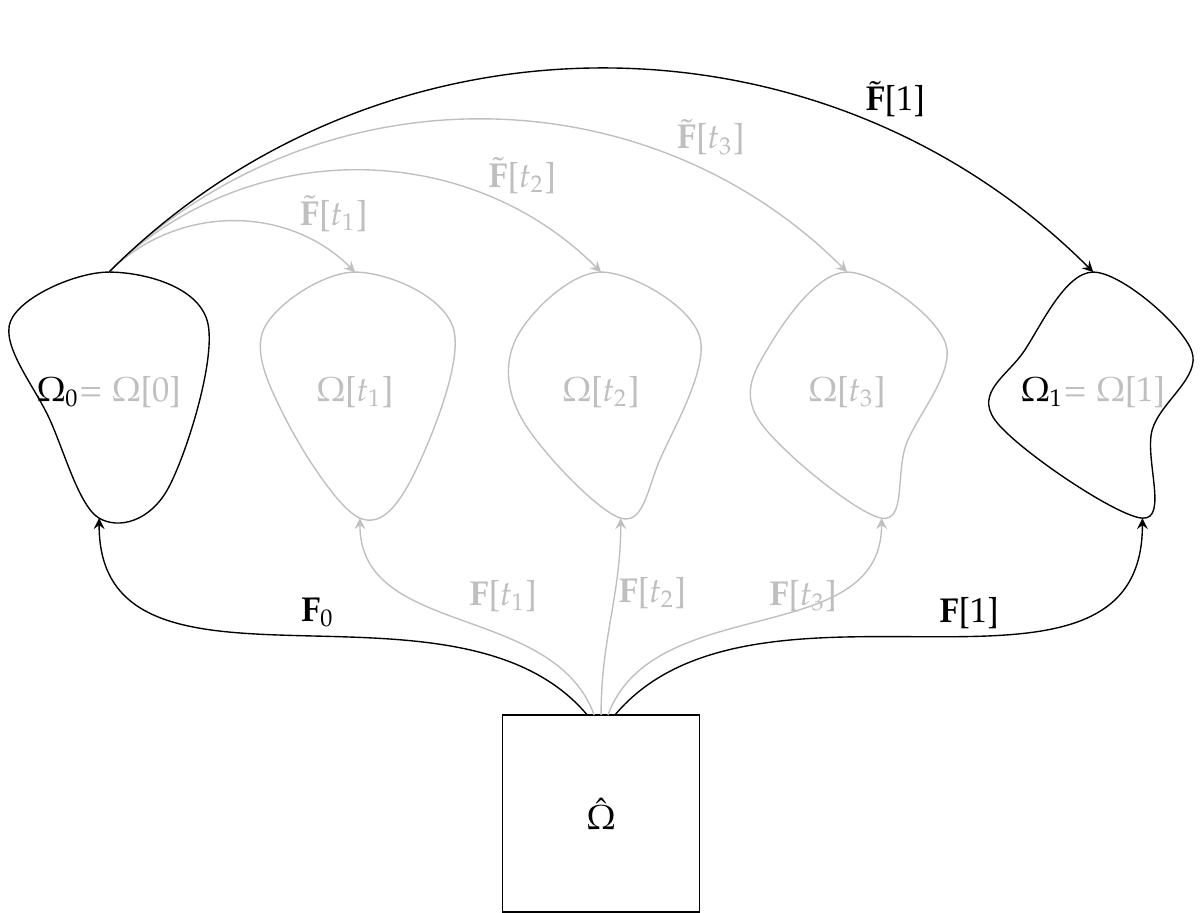}
        \caption{Illustration of shape morphing: A physical domain $\Omega_0$, which is the image of a reference domain $\hat \Omega$ under transformation $\map_0$, is continuously transformed into $\Omega_1 = \maptilF[1](\Omega_0)$. For any $t \in [0,1]$ the intermediate domain is given by $\Omega[t] = \maptilF[t](\Omega_0) = \map[t](\hat \Omega)$ with $\map[t]= \maptilF[t] \circ \map_0$.}
        \label{fig_transformations}
    \end{figure}

    \subsection{Sensitivity of system matrices of transformed domains} \label{sec_abstract_trafo}
    
    We show how sensitivities of stiffness/mass matrices on a parameter-dependent domain $\Omega[t] = \maptil[t](\Omega)$ can be computed on a reference domain $\Omega$, see Figure~\ref{fig_transformations}. This can be applied either for the setting
     \begin{align*}
         \Omega:= \Omega_0, \quad \maptil[t] := \maptilF[t], 
    \qquad \mbox{ or } \qquad
        \Omega:= \hat \Omega, \quad \maptil[t] := \map[t] = \maptilF[t]\circ \map_0,
    \end{align*}
    where the mapping $\maptil[t]$ inherits the positivity assumption on $\map$, i.e., $\det\left(\partial\maptil[t]\right) > 0$.
      
    The following lemmas state that the stiffness and mass matrices for the physical domain $\Omega[t]$ can be written in terms of integrals over the fixed domain~$\Omega$.   
    We begin with the case of $\Hone$.
    \begin{lemma} \label{lem_stiffmat_H1}
            Let $t \geq 0$ fixed and let $\stiffgrad[\parat]$ be the stiffness matrix and $\massgrad[\parat]$ the mass matrix of the Laplacian as defined in \eqref{eq_stiffH1} on the domain $\Omega[t] = \maptil[t](\Omega)$, i.e.,
            \begin{align} \label{eq_Kt_Mt}
                \stiffgrad[t]_{i,j} = \int_{\Omega[t]} \nabla \fembasis_{t,j} \cdot \nabla \fembasis_{t,i} \, \dx,
                \qquad
                \massgrad[t]_{i,j} = \int_{\Omega[t]} \fembasis_{t,j} \, \fembasis_{t,i} \, \dx 
            \end{align}
            for $i, j \in \{1, \dots, \Ndof\}$ with basis functions $\{\fembasis_{t,i}\}_{i=1}^\Ndof$.            
            Let $\fembasis_i := \fembasis_{t,i} \circ \maptil[t]$. Then it holds
            \begin{align}
                \stiffgrad[t]_{i,j} = \int_{\Omega} \mathbf A[t] \nabla \fembasis_j \cdot \nabla \fembasis_i \, \dx, \qquad      \massgrad[t]_{i,j} = \int_{\Omega} \det (\partial \maptil[t]) \fembasis_j \cdot \fembasis_i \, \dx
            \end{align}
            with 
            \begin{align} \label{eq_defAt}
                \mathbf A[t] := \det(\partial \maptil[t]) \partial \maptil[t]^{-1} \partial \maptil[t]^{-\top} .
            \end{align}
    \end{lemma}
    \begin{proof} 
    By a change of variables $\mathbf{y} = \maptil[t](\mathbf{x})$ in \eqref{eq_Kt_Mt} and Lemma \ref{lem_trafos}(a) we obtain for the stiffness matrix 
	\begin{align*}
        \stiffgrad[t]_{i,j} =& \int_{\Omega[t]} \nabla \fembasis_{t,j} \cdot  \nabla \fembasis_{t,i} \; \dy \\
        =& \int_{\Omega} \left(\nabla \fembasis_{t,j}\right)\circ \maptil[t] \cdot \left(\nabla \fembasis_{t,i}\right)\circ \maptil[t] \, \det(\partial \maptil[t]) \; \dx \\
        =& \int_{\Omega} \left(  \partial \maptil[t]^{-\top} \nabla \fembasis_j \right) \cdot \left( \partial \maptil[t]^{-\top}  \nabla \fembasis_i \right) \det(\partial \maptil[t]) \, \dx\\
        =& \int_{\Omega} \mathbf A[t] \nabla \fembasis_j \cdot \nabla \fembasis_i \, \dx.
	\end{align*}
	The result for the mass matrix follows straightforwardly by the same coordinate transformation.
    \end{proof}
    Next, we state the corresponding result in the case of $\Hcurl$.
    \begin{lemma} \label{lem_stiffmat_Hcurl}
        Let $t \geq 0$ fixed and let $\stiff[\parat]$ be the stiffness matrix and $\mass[\parat]$ the mass matrix of the electromagnetic problem as defined in \eqref{eq_stiffHcurl} on the domain $\Omega[t] = \maptil[t](\Omega)$, i.e.,
        \begin{align}
            \stiff[t]_{i,j} = \int_{\Omega[t]} \curl \vec{\fembasis}_{t,j} \cdot \curl \vec{\fembasis}_{t,i} \, \dx,
            \qquad
            \mass[t]_{i,j} = \int_{\Omega[t]} \vec{\fembasis}_{t,j} \cdot \vec{\fembasis}_{t,i} \, \dx 
        \end{align}
        for $i, j \in \{1, \dots, \Ndof\}$ with basis functions $\{\vec{\fembasis}_{t,i}\}_{i=1}^\Ndof$.            
        Let $ \vec{\fembasis}_i := \partial \maptil[t]^\top \vec{\fembasis}_{t,i} \circ \maptil[t]$. Then it holds
        \begin{align}
            \stiff[t]_{i,j} = \int_{\Omega} \mathbf C[t] \curl \vec{\fembasis}_j \cdot \curl \vec{\fembasis}_i \, \dx, \qquad      \mass[t]_{i,j} = \int_{\Omega} \mathbf A[t] \vec{\fembasis}_j \cdot \vec{\fembasis}_i \, \dx
        \end{align}
        with $\mathbf A[t]$ as defined in \eqref{eq_defAt} and 
        \begin{align} \label{eq_defCt}
             \mathbf C[t] := \frac{1}{\det(\partial \maptil[t])}\partial \maptil[t]^\top \partial \maptil[t].
        \end{align}
    \end{lemma}
    \begin{proof}
    A change of variables $\mathbf{y} = \maptil[t](\mathbf{x})$ and Lemma \ref{lem_trafos}(b) yields for the stiffness matrix 
	\begin{align*}
        \stiff[t]_{i,j} =& \int_{\Omega[t]} \curl \vec{\fembasis}_{t,j} \cdot \curl \vec{\fembasis}_{t,i} \dy \\
        =& \int_{\Omega} \left(\curl \vec{\fembasis}_{t,j}\right)\circ \maptil[t] \cdot \left(\curl  \vec{\fembasis}_{t,i}\right)\circ \maptil[t] \, \det(\partial \maptil[t]) \dx \\
        =& \int_{\Omega} \left( \frac{1}{\det(\partial \maptil[t])} \partial \maptil[t] \curl \vec{\fembasis}_j \right) \cdot \left( \frac{1}{\det(\partial \maptil[t])} \partial \maptil[t] \curl \vec{\fembasis}_i \right) \det(\partial \maptil[t]) \, \dx\\
        =& \int_{\Omega}\mathbf C[t] \curl \vec{\fembasis}_j \cdot \curl \vec{\fembasis}_i \, \dx.
	\end{align*}
	For the mass matrix, we obtain with the coordinate transformation $\mathbf{y} = \maptil[t](\mathbf{x})$
	\begin{align*}
        \mass[t]_{i,j} =& \int_{\Omega[t]} \vec{\fembasis}_{t,j} \cdot \vec{\fembasis}_{t,i} \dy \\
        =& \int_{\Omega} \vec{\fembasis}_{t,j}\circ \maptil[t] \cdot \vec{\fembasis}_{t,i}\circ \maptil[t] \, \det(\partial \maptil[t]) \dx \\
        =& \int_{\Omega} \left(\partial \maptil[t]^{-\top} \vec{\fembasis}_j \right) \cdot \left(\partial \maptil[t]^{-\top} \vec{\fembasis}_i \right) \det(\partial \maptil[t]) \, \dx\\
        =& \int_{\Omega} \mathbf A[t]  \vec{\fembasis}_j \cdot  \vec{\fembasis}_i \, \dx.
	\end{align*}
    \end{proof}

    Lemmas \ref{lem_stiffmat_H1} and \ref{lem_stiffmat_Hcurl} allow us to compute the derivatives of the stiffness and mass matrices with respect to the shape parameter $\parat$ by simply differentiating the matrices $\mathbf A[t]$ and $\mathbf C[t]$ defined in \eqref{eq_defAt} and \eqref{eq_defCt}, respectively.
    
    The following result can be found, e.g., in \cite{Sturm_2015ab}.
	\begin{lemma} \label{lem_derivJacobi}
        Let $\maptil[t]$ a smooth transformation with Jacobian $\partial \maptil[t]$, and $\mathbf A[t]$ be as defined in \eqref{eq_defAt}. Then it holds
        \begin{align}
            \ddt \left( \left(\partial \maptil[t]\right)^{-1} \right) =& - (\partial \maptil[t])^{-1} \left(\ddt\partial \maptil[t]\right) (\partial \maptil[t])^{-1}, \label{eq_ddt_partialFinv}\\
            \ddt \det(\partial \maptil[t]) =& \tr\left(\left(\ddt\partial \maptil[t]\right) \partial \maptil[t]^{-1} \right) \det(\partial \maptil[t]), \label{eq_ddt_det}
        \end{align}
        and
        \begin{align}
            \begin{aligned} \label{eq_ddt_A}
            \ddt \mathbf A[t] =&  \tr \left( \left(\ddt\partial \maptil[t] \right) \partial \maptil[t]^{-1} \right) \mathbf A[t] - (\partial \maptil[t])^{-1} \left(\ddt\partial \maptil[t] \right) \mathbf A[t] \\
            &- \left( (\partial \maptil[t])^{-1} \left(\ddt\partial \maptil[t]\right) \mathbf A[t] \right)^\top.
            \end{aligned}
        \end{align}
    \label{lem:trafos}
	\end{lemma}
    For sake of completeness, a proof is given in \ref{apdx:proof_dAdt}. Using the same tools, we can also obtain the derivative of the matrix $\mathbf C[t]$.	
	\begin{lemma}
	Let $\maptil[t]$ a smooth transformation with Jacobian $\partial \maptil[t]$, and $\mathbf C[t]$ be as defined in \eqref{eq_defCt}. Then it holds
	    \begin{align} \label{eq_defddtCt}
	        \begin{aligned}
	        \ddt \mathbf C[t] =& - \tr \left(\left(\ddt\partial \maptil[t]\right) \partial \maptil[t]^{-1} \right) \mathbf C[t] + \frac{1}{\det(\partial \maptil[t])} \left(\ddt\partial \maptil[t]\right)^\top \partial \maptil[t]  \\
	        &+ \frac{1}{\det(\partial \maptil[t])}  \left( \left(\ddt\partial \maptil[t]\right)^\top \partial \maptil[t] \right)^\top.
	        \end{aligned}
	   \end{align}
	\end{lemma}
	\begin{proof}
	    By means of the chain rule, we have
        \begin{align*}
            \ddt \mathbf C[t] =& \ddt \left( \frac{1}{\det(\partial \maptil[t])}\partial \maptil[t]^\top \partial \maptil[t] \right) \\
            =& \ddt \left( \frac{1}{\det(\partial \maptil[t])} \right) \partial \maptil[t]^\top \partial \maptil[t] + \frac{1}{\det(\partial \maptil[t])} \ddt \left( \partial \maptil[t]^\top \right) \partial \maptil[t]  \\
            &+ \frac{1}{\det(\partial \maptil[t])} \partial \maptil[t]^\top  \ddt \left( \partial \maptil[t] \right)
        \end{align*}
        For the first term we get by \eqref{eq_ddt_det} and the chain rule
        \begin{align*}
            \ddt \left( \frac{1}{\det(\partial \maptil[t])} \right) \partial \maptil[t]^\top \partial \maptil[t] =& - \frac{\ddt\det(\partial \maptil[t])}{\det(\partial \maptil[t])^2} \partial \maptil[t]^\top \partial \maptil[t] \\
            =&- \frac{\det(\partial \maptil[t])}{\det(\partial \maptil[t])^2} \tr \left(\left(\ddt\partial \maptil[t]\right) \partial \maptil[t]^{-1} \right) \partial \maptil[t]^\top \partial \maptil[t] \\
            =&- \tr\left(\left(\ddt\partial \maptil[t]\right) \partial \maptil[t]^{-1} \right) \mathbf C[t],
        \end{align*}
        which yields the result.
	\end{proof}
	
	\begin{corollary} \label{cor_dAC}
	    For a transformation $\maptil[t]$ of the form $\maptil[t](\mathbf{x}) =  \mathbf{x} + t \vec{V}(\mathbf{x})$ with a smooth vector field $\vec{V}$, it holds $\partial \maptil[t](\mathbf{x}) =  \mathbf{I} + t \partial \vec{V}(\mathbf{x})$ and thus
	    \begin{align}
	     \ddt \det(\partial \maptil[t]) =& \tr\left(\partial \vec{V} \partial \maptil[t]^{-1} \right) \det(\partial \maptil[t]), \label{eq_dDet}\\
	        \ddt \mathbf A[t] =&  \tr(\partial \vec{V} \partial \maptil[t]^{-1} ) \mathbf A[t] - (\partial \maptil[t])^{-1} \partial \vec{V} \, \mathbf A[t]  - ( (\partial \maptil[t])^{-1} \partial \vec{V} \, \mathbf A[t])^\top, \label{eq_dA}\\
	        \ddt \mathbf C[t] =& - \tr(\partial \vec{V} \partial \maptil[t]^{-1} ) \mathbf C[t] + \frac{1}{\det(\partial \maptil[t])} \partial \vec{V}^\top \partial \maptil[t]  + \frac{1}{\det(\partial \maptil[t])}  (\partial \vec{V}^\top \partial \maptil[t] )^\top.  \label{eq_dC}
	    \end{align}
	\end{corollary}
	
	\subsection{Application to shape morphing} \label{sec_morph}
	We apply the results of Section \ref{sec_abstract_trafo} to the setting illustrated in Figure \ref{fig_transformations} so as to derive formulas for the sensitivities of the mass and stiffness matrices in the case of the Laplacian and electromagnetics on the fixed physical domain $\Omega_0$ (i.e., we set $\Omega:= \Omega_0$ and $\maptil[t]:=\maptilF[t]$) and on the reference domain $\hat \Omega$ (i.e., $\Omega:= \hat \Omega$ and $\maptil[t]:=\map[t] = \maptilF[t]\circ \map_0$).
	Recall the bases $\{\hat{\fembasis}_i\}_{i = 1}^{\Ndof}$ and $\{\vec{\hat{\fembasis}}_i\}_{i = 1}^{\Ndof}$ of the discrete scalar and vector-valued spaces $\hat{W}(\hat{\Omega})$ and $\hat{\textbf{W}}(\hat{\Omega})$, respectively. Moreover, recall that it holds for the fixed physical domain $\Omega_0 = \map_0(\hat \Omega)$. Thus, it holds that the functions $\fembasis_{0,i} :=  \hat \fembasis_i \circ \map_0^{-1}$ and $\vec{\fembasis}_{0,i} := (\partial  \map_0)^{-\top} ( \vec{\hat \fembasis}_i \circ \map_0^{-1})$ form bases of $W[0](\Omega_0)$ and $\textbf{W}[0](\Omega_0)$, respectively (see \eqref{eq_def_Wt} and \eqref{eq_def_Wtvec}).
	
	Given two sets of control points $\{\mathbf{P}_{0,i}\}_{i = 1}^{\Ngeo}$ and $\{\mathbf{P}_{1,i}\}_{i = 1}^{\Ngeo}$ corresponding to an original domain $\Omega_0$ and a target domain $\Omega_1$, respectively, {represented using a spline basis of dimension \Ngeo}, we are interested in the stiffness and mass matrices for intermediate domains represented by a scalar parameter $t \in [0,1]$ as well as their derivatives with respect to $t$. 
    {Note, that it is possible but not necessary to choose the spaces such that the dimensions for the solution space $\Ndof$ and the geometry space $\Ngeo$ are the same.}
	
	\paragraph{Calculations on fixed physical domain $\Omega_0$}
	The transformation $\maptilF[t]$ defined by
	\begin{align}
	    \maptilF[t](\mathbf{x}) = \sum_{i =1}^{\Ngeo} ( \mathbf{P}_{0,i} + t(\mathbf{P}_{1,i} - \mathbf{P}_{0,i})) \fembasis_{0,i}(\mathbf{x}), \; \mathbf{x} \in \Omega_0
	\end{align}
	for $t \in [0,1]$ represents a smooth transition from $\Omega_0 = \maptilF[0](\Omega_0)$ to $\Omega_1 = \maptilF[1](\Omega_0)$, see also Figure \ref{fig_transformations}. Moreover, note that the transformation $\maptilF[t]$ is of the form $\maptilF[t](\mathbf{x}) = \mathbf{x} + t \vec{V}(\mathbf{x})$ with
	\begin{align*}
	    \vec{V}(\mathbf{x}) := \sum_{i =1}^{\Ngeo} (\mathbf{P}_{1,i} - \mathbf{P}_{0,i}) \fembasis_{0,i}(\mathbf{x}), \; \mathbf{x} \in \Omega_0.
	\end{align*}

    We define $\Omega[t] := \maptilF[t](\Omega_0)$. Then we obtain for the stiffness and mass matrices on $\Omega[t]$
	\begin{align}
        \ddt \stiffgrad[t]_{i,j} =&  \int_{\Omega_0} \ddt \mathbf A[t]  \nabla \fembasis_{0,j} \cdot \nabla \fembasis_{0,i} \, \dx, \\
        \ddt \massgrad[t]_{i,j} =&  \int_{\Omega_0} \ddt \det(\partial \maptilF[t]) \fembasis_{0,j} \, \fembasis_{0,i} \, \dx,
	\end{align}
	in the case of the Laplacian, and
	\begin{align}
        \ddt \stiff[t]_{i,j} =&  \int_{\Omega_0} \ddt \mathbf C[t]  \curl \vec{\fembasis}_{0,j} \cdot \curl \vec{\fembasis}_{0,i} \, \dx,\\
        \ddt \mass[t]_{i,j} =&  \int_{\Omega_0} \ddt \mathbf A[t] \vec{\fembasis}_{0,j} \cdot \vec{\fembasis}_{0,i} \, \dx
	\end{align}
	in the case of electromagnetism.
	Here, the matrices $\mathbf A[t]$ and $\mathbf C[t]$ involve the transformation $\maptil[t] = \maptilF[t]$ and their derivatives can be evaluated by means of Corollary \ref{cor_dAC} using
	\begin{align*}
	    \partial \vec{V}(\mathbf{x}) &= \sum_{i =1}^{\Ngeo} (\mathbf{P}_{1,i} - \mathbf{P}_{0,i}) (\nabla \fembasis_{0,i}(\mathbf{x}))^\top \; \in \R^{\ndim \times \ndim}.
	\end{align*}
	
    In the following section, we emphasize the spatial variable of the domain with respect to which the differentiation is carried out by adding it as an index to the $\partial$-operator.
	
	\paragraph{Calculations on reference domain $\hat \Omega$}
	The derivatives of stiffness and mass matrices can as well be computed on the reference domain $\hat \Omega$, e.g., for electromagnetics as
	\begin{align*}
        \ddt \stiff[t]_{i,j} =&  \int_{\hat \Omega} \ddt \mathbf{ \hat C}[t]  \curl \vec{\hat \fembasis}_j \cdot \curl \vec{\hat \fembasis}_i \, \dx,\\
        \ddt \mass[t]_{i,j} =&  \int_{\hat \Omega} \ddt \mathbf{\hat A}[t]  \vec{\hat \fembasis}_j \cdot  \vec{\hat \fembasis}_i \, \dx
	\end{align*}
	where $\mathbf{ \hat C}[t]$, $\mathbf{ \hat A}[t]$ are as defined in \eqref{eq_defCt} and \eqref{eq_defAt}, respectively, using the transformation $\maptil[t] = \map[t]$. Note that $\map[t](\mathbf{\hat x}) = (\maptilF[t] \circ \map_0)(\mathbf{\hat x}) = \map_0(\mathbf{\hat x}) + t \vec{V}(\map_0(\mathbf{\hat x}))$ and thus
	\begin{align}
	    \partial_{\mathbf{\hat x}} \map[t](\mathbf{\hat x}) 
	    =& \partial_{\mathbf{\hat x}} \map_0(\mathbf{\hat x}) + t \partial_{\mathbf{x}} \vec{V}(\map_0(\mathbf{\hat x})) \, \partial_{\mathbf{\hat x}} \map_0(\mathbf{\hat x}) 
	\end{align}
	where, using $\fembasis_{0,i}\circ \map_0 = \hat{\fembasis}_i $,
	\begin{align*}
	    \partial_{\mathbf{x}} \vec{V}(\map_0(\mathbf{\hat x}) ) =& \partial_{\mathbf{x}} \left( \sum_{i =1}^{\Ngeo} (\mathbf{P}_{1,i} - \mathbf{P}_{0,i}) \fembasis_{0,i}( \map_0(\mathbf{\hat x})) \right)\\
	    =& \partial_{\mathbf{x}} \left( \sum_{i = 1}^{\Ngeo} (\mathbf{P}_{1,i} - \mathbf{P}_{0,i}) \hat{\fembasis}_i \circ \map_0^{-1}(\map_0(\mathbf{\hat x})) \right) \\
	    =&\sum_{i = 1}^{\Ngeo} (\mathbf{P}_{1,i} - \mathbf{P}_{0,i}) \partial_{\mathbf{\hat x}} \hat{\fembasis}_i(\mathbf{\hat x}) \, \partial_{\mathbf{x}}(\map_0^{-1})(\map_0(\mathbf{\hat x})) \\
	    =&\sum_{i = 1}^{\Ngeo} (\mathbf{P}_{1,i} - \mathbf{P}_{0,i}) \partial_{\mathbf{\hat x}} \hat{\fembasis}_i(\mathbf{\hat x}) \, (\partial_{\mathbf{\hat x}} \map_0(\mathbf{\hat x}) )^{-1}.
	\end{align*}
	Thus, we have
	\begin{align*}
	    \partial_{\mathbf{\hat x}} \map[t](\mathbf{\hat x}) =& \partial_{\mathbf{\hat x}} \map_0(\mathbf{\hat x}) + t \sum_{i = 1}^{\Ngeo} (\mathbf{P}_{1,i} - \mathbf{P}_{0,i}) \partial_{\mathbf{\hat x}} \hat{\fembasis}_i(\mathbf{\hat x}) \, (\partial_{\mathbf{\hat x}} \map_0(\mathbf{\hat x}) )^{-1} \partial_{\mathbf{\hat x}} \map_0(\mathbf{\hat x}) \\
	     =& \partial_{\mathbf{\hat x}} \map_0(\mathbf{\hat x}) + t \sum_{i = 1}^{\Ngeo} (\mathbf{P}_{1,i} - \mathbf{P}_{0,i}) \partial_{\mathbf{\hat x}} \hat{\fembasis}_i(\mathbf{\hat x}) 
	\end{align*}
\newcommand{\G}{\mathbf{G}}
\renewcommand{\P}{\mathbf{P}}
\newcommand{\x}{\mathbf{x}}
\newcommand{\hatx}{\mathbf{\hat{x}}}
\newcommand{\dG}{\partial \G[t]}
\newcommand{\dGinv}{\partial \G[t]^{-1}}
\newcommand{\dV}{\partial \mathbf{V}}

\newcommand{\funi}{\frac{1}{\det(\dG)}}
\newcommand{\fid}{- \frac{1}{\det(\dG)} \tr(\dV \dG^{-1} )}
\newcommand{\fidd}{\frac{1}{\det(\dG)} \left[ \tr(\dV \dG^{-1} ) ^2 + \tr\left[(\dV \dGinv) (\dV \dGinv)\right]\right]}
\newcommand{\fii}{\dG^\top}
\newcommand{\fiid}{\dV^\top}
\newcommand{\fiii}{\dG}
\newcommand{\fiiid}{\dV}
\newcommand{\f}[2]{f_{#1}^{(#2)}}

\subsection{Higher order derivatives}
For computing higher order derivatives of the system matrices in the case of electromagnetism exemplarily, we accordingly differentiate the terms $\mathbf C[t]$ and $\mathbf A[t]$.
{Let us look at further derivatives of the term $\ddt \mathbf C[t]$, which work analogously for~$\ddt \mathbf A[t]$.}

\begin{lemma} \label{lem:d2dt2C}
        Let $\maptil[t]{(\x)=\x+t\mathbf{V}(\x)}$ a smooth transformation {for a given vector field $\mathbf{V}$} with Jacobian $\partial \maptil[t]$, and $\mathbf C[t]$ as defined in \eqref{eq_defCt} and $\ddt \mathbf C[t]$ as defined in \eqref{eq_defddtCt}. Then it holds
        
        \begin{align*}
            \ddt \left[ \frac{1}{\det(\dG)} \right] &= - \frac{1}{\det(\dG)} \tr(\dV \dG^{-1} ),\\
            \frac{\mathrm{d}^2}{\mathrm{d}t^2} \left[ \frac{1}{\det(\dG)} \right] &= \frac{1}{\det(\dG)} \left[ \tr(\dV \dG^{-1} ){^2} + \tr\left[(\dV \dGinv) (\dV \dGinv)\right]\right],
        \end{align*}
        and
        \begin{align*}
            \frac{\mathrm{d}^2}{\mathrm{d}t^2} \mathbf C[t] &= \tr\left[(\dV \dGinv) (\dV \dGinv)\right] \mathbf C[t] -  \tr(\dV \dGinv ) \ddt \mathbf C[t]\\
            & - \frac{1}{\det(\dG)} \tr(\dV \dG^{-1} ) \cdot \left[ \dV^\top \dG + \dG^\top  \dV \right]\\
            & {+ \frac{2}{\det(\dG)} \cdot  \dV ^\top \dV }.
        \end{align*}
\end{lemma}
The proof for this result can be found in \ref{apdx:proof_d2Cdt2}.
This allows the computation of the second derivative of the stiffness matrix via Lemma~\ref{lem_stiffmat_Hcurl}. 
For the computation of the third derivative, we use the following Lemma.
\begin{lemma}\label{lem:d3dt3C}
        Let $\maptil[t]{(\x)=\x+t\mathbf{V}(\x)}$ a smooth transformation {for a given vector field $\mathbf V$} with Jacobian $\partial \maptil[t]$, and $\mathbf C[t]$ be as defined in \eqref{eq_defCt}. Then it holds
        \begin{align*}
            \frac{\mathrm{d}^3}{\mathrm{d}t^3} \mathbf C[t] =& \biggl(-\tr(\dV\dGinv)\cdot \fidd   \\
            &-2 \cdot \frac{1}{\det(\dG)} \cdot \biggl[ \tr(\dV \dGinv)\cdot \tr(\dV \dGinv \dV \dGinv) \\
            &+ \tr (\dV \dGinv \dV \dGinv\dV \dGinv) \biggr]\biggr) \fii  \fiii \\
            &- 6 \cdot \frac{1}{\det(\dG)} \tr(\dV \dG^{-1} ) \fiid \fiiid \\
            &+ 3 \cdot \fidd \fiid \fiii\\
            &+ 3 \cdot \fidd \fii \fiiid.
        \end{align*}
\end{lemma}
A proof, based on the General Leibniz rule, is presented in \ref{apdx:proof_dnCdtn}.

In order to compute also higher order derivatives of $\mathbf C[t]$ and analogously for $\mathbf A[t]$, we use symbolic differentiation for~$\mathbf C[t]$ and~$\mathbf A[t]$ using the Symbolic Math Toolbox~\cite{mathworks_2022aa} in MATLAB\textsuperscript{\textregistered} which only requires a few lines of code.
In essence, we state the Jacobian $\dG(\mathbf{x}) = \mathbf{I} + t \cdot \dV(\mathbf{x})$ and the corresponding term~\eqref{eq_defCt} for $\mathbf C[t]$ or~\eqref{eq_defAt} for $\mathbf A[t]$.
In a loop, we can then iteratively differentiate $\mathbf C[t]$ and $\mathbf A[t]$ up to the desired order.
This allows us to perform analytical differentiation easily up to high orders and also the computation of the system matrix derivatives in closed-form \cite{Ziegler_2022ac}.

\section{Numerical Examples}\label{sec:numerics}
In the following subsections, we present two applications in which we make use of our results.

\subsection{Uncertainty quantification for a pillbox cavity of uncertain radius}
\begin{figure}
\begin{subfigure}[t]{0.45\textwidth}
\includegraphics[scale=1]{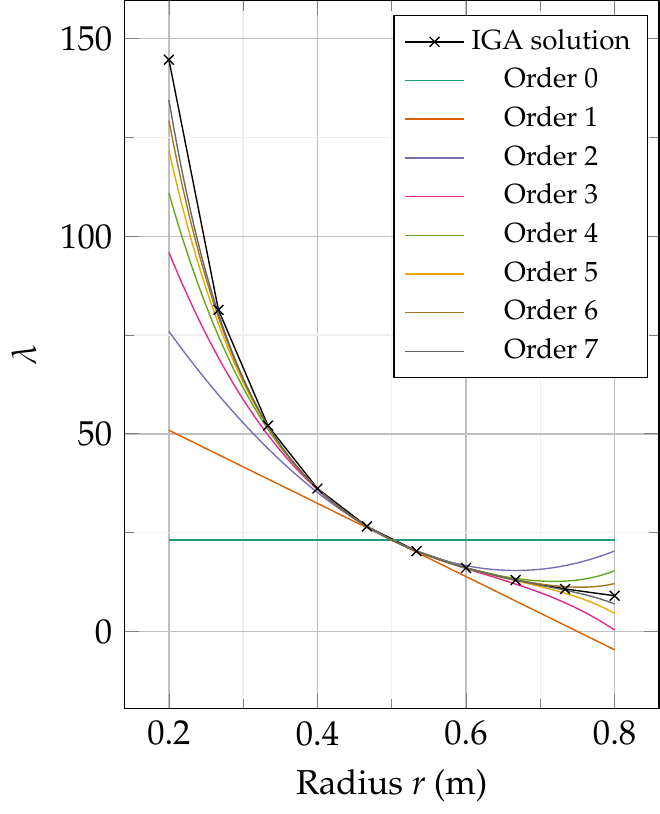}
\caption{Eigenfrequencies of the pillbox cavity estimated by the Taylor series expansion at $r_d=\SI{0.5}{\meter}$ using IGA and the matrix derivatives; black crosses denote the result of independent eigenvalue solver calls.}
\label{fig:uq_pillbox_iga}
\end{subfigure}\hfill
\begin{subfigure}[t]{0.45\textwidth}
\includegraphics[scale=1]{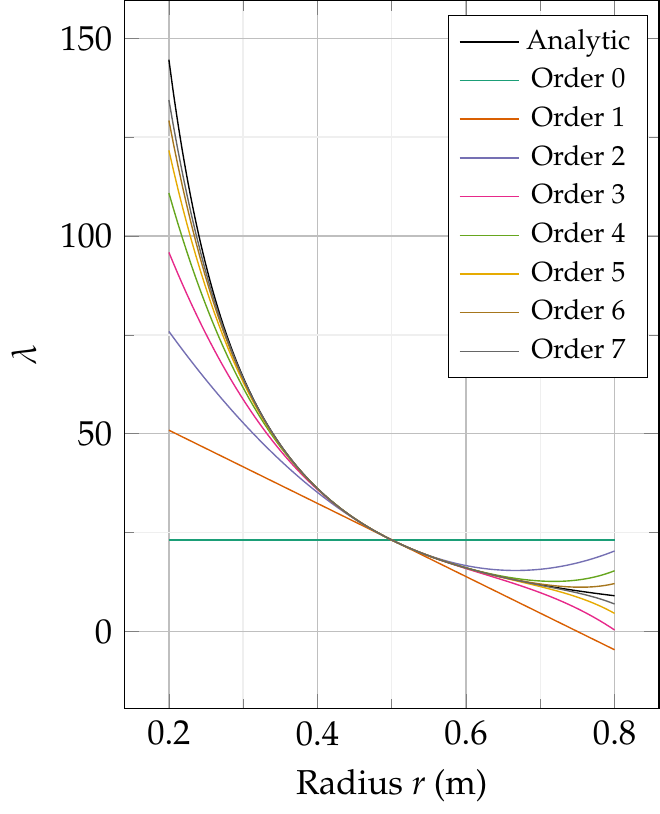}
\caption{Eigenfrequencies of the pillbox cavity estimated using the Taylor series expansion of the closed-form solution~\eqref{eq:tm010frequ} at $r_d=\SI{0.5}{\meter}$.}
\label{fig:uq_pillbox_closedform}
\end{subfigure}
\caption{Comparison of eigenvalues and the Taylor expansions at $r_d = \SI{0.5}{\meter}$.}
\label{fig:uq_pillbox}
\end{figure}

As a first benchmark example, we consider the setting introduced in~\cite{Corno_2015ac} and investigate the resonant frequency of the fundamental mode of a pillbox cavity.
We assume (unrealistically large) uncertainties of the radius $r$ and investigate the variance in the frequency with respect to the radius around the design value $r_d = \SI{0.5}{\meter}$.
We compare our solutions to the closed-form eigenvalue of the fundamental mode
\begin{equation}
    \lambda(r) = \frac{x_{01}^2}{r^2}
    \label{eq:tm010frequ}
\end{equation}
where $x_{01}$ is the first root of the Bessel function $J_0(x)$ \cite{Abramowitz_1972aa}. We assume a uniformly distributed radius $\theta_r \sim \mathcal{U}(a,b)$ where $a=\SI{0.2}{\meter}$ and $b=\SI{0.8}{\meter}$ 
and compare Taylor expansions of orders zero to seven at the design radius in the uncertain interval in Fig.~\ref{fig:uq_pillbox}.

To investigate the uncertainties, we use the mappings $\mathbf{F}_0$ and $\mathbf{F}_1$ which describe the domains corresponding to pillboxes of radius $r=a$ and $r=b$, respectively.
The two mappings are given by control points $\mathbf{P}_{0,i}$ and $\mathbf{P}_{1,i}$ of compatible dimensions and $i\in \{1,\ldots,\Ngeo\}$, which we assemble in the convex combination
\begin{equation}\label{eq:Ptconvexcomb}
    \mathbf{P}_{i}[t] = (1-t) \mathbf{P}_{0,i} + t \mathbf{P}_{1,i}.
\end{equation}
The NURBS-related weights $w_i$ from \eqref{eq:NURBS} do not change with respect to the deformation parameter $t$ in this example. The exact definitions can be found in~\cite{Ziegler_2022ac}.

In Fig.~\ref{fig:uq_pillbox_iga} we show the simulation results based on IGA and the Taylor series expansions. Let us recall that the Taylor expansions require the solution of only a single eigenvalue problem and several linear systems {\eqref{eq:eigprobderiv}} with the same system matrix. For reference, we add several solutions of a conventional eigenvalue solver based on the IGA model and mark them with black crosses. For comparison, we follow the same approach in Fig.~\ref{fig:uq_pillbox_closedform}, but here perform the expansions based on the closed-form formula~\eqref{eq:tm010frequ} and obtain almost identical curves. 

To illustrate the practical relevance of the Taylor series expansion, we compute the expected value of the fundamental mode. The closed-form solution is
\begin{equation}
    \label{eq:expected_value}
    \mathbb{E}(\lambda(\theta_r)) = \int_{-\infty}^{\infty} \lambda(\theta_r) \rho(\theta_r) \mathrm{d}\theta_r = \frac{x_{01}^2}{b-a} \left(\frac{1}{a}-\frac{1}{b}\right)
\end{equation}
with the density $\rho(\theta_r)$. Now we compute the expected values using Taylor expansions of the closed-form and numerical solution. Classical methods from uncertainty quantification, e.g., stochastic collocation \cite{Xiu_2010aa,Georg_2019aa}, would solve \eqref{eq:expected_value} by numerical quadrature and possibly a mode matching (see below). Here, this is not necessary, the integral is exactly evaluated by integrating the corresponding polynomials. In Fig.~\ref{fig:uq_pillbox_relError} we plot the relative errors with respect to the closed-form solution
\begin{equation*}
    \mathbb{E}(\lambda(r = \SI{0.5}{\meter})) \approx 36.1449
\end{equation*}
for Taylor expansions up to order fourteen and observe similar convergence rates in both cases.

\begin{figure}
\centering
\includegraphics[scale=1]{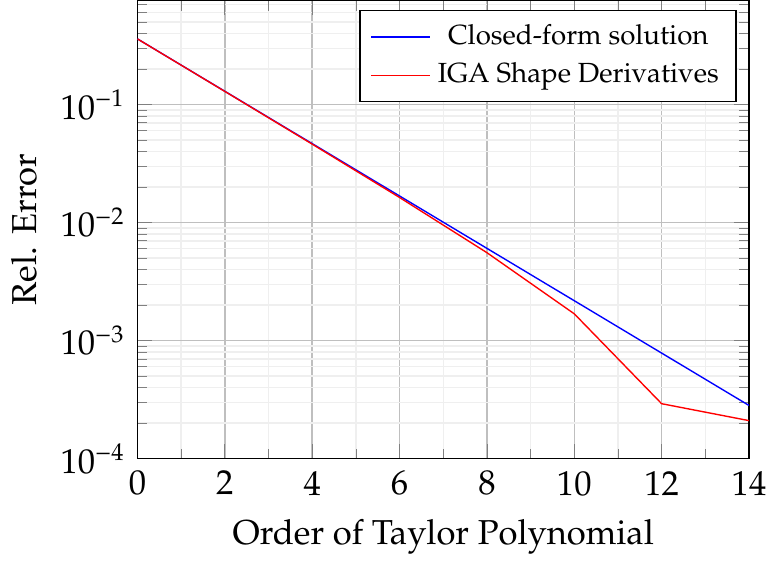}
\caption{Convergence of $\mathbb{E}(\lambda(\theta_r))$ for the Taylor expansion around $r_d$ based on the closed-form solution and the IGA shape derivatives for polynomials of degrees zero to fourteen.}
\label{fig:uq_pillbox_relError}
\end{figure}

\subsection{Eigenvalue estimation for shape morphing}

We apply the shape derivatives for the eigenmode tracking described in~\cite{Ziegler_2022ab}, where a TESLA cavity is morphed to a pillbox and the eigenmodes are tracked along the shape deformation for an automatic recognition of the eigenmodes in the TESLA cavity.
In order to properly follow the eigenmodes along the deformation, a matching procedure is necessary to identify consistent modes.
For the purpose of the multi-step method in~\cite{Ziegler_2022ab}, we estimate the eigenmode at the next morphing step by first order Taylor expansion and match the estimated eigenmode and candidate solutions based on a correlation factor~\cite{Jorkowski_2018aa}.
In this work, we propose employing higher order derivatives with respect to the shape deformation to improve the eigenmode estimations via higher order Taylor expansion, in order to increase the confidence in the eigenmode matching.
We exemplify this for the fundamental (accelerating) mode of the 9-cell TESLA cavity in Fig.~\ref{fig:trackingPlot}.

Analogously to above, we construct the two mappings~$\mathbf{F}_0 \coloneqq \mathbf{F}[t =0]$ and~$\mathbf{F}_1 \coloneqq \mathbf{F}[t =1]$, where~$\mathbf{F}_0$ maps the reference domain on the TESLA cavity and~$\mathbf{F}_1$ to the pillbox cavity with control points~$\mathbf{P}_{0,i}$ and~$\mathbf{P}_{1,i}$, respectively, of compatible dimensions.
The parametrized mapping is obtained by~\eqref{eq:Ptconvexcomb}. 
In this example, we keep the NURBS-related weights $w_i$ constant along the shape morphing. 
The definitions are stored in the repository~\cite{Ziegler_2022ac}.
\begin{figure}[htbp]
    \centering
    \includegraphics[scale=1]{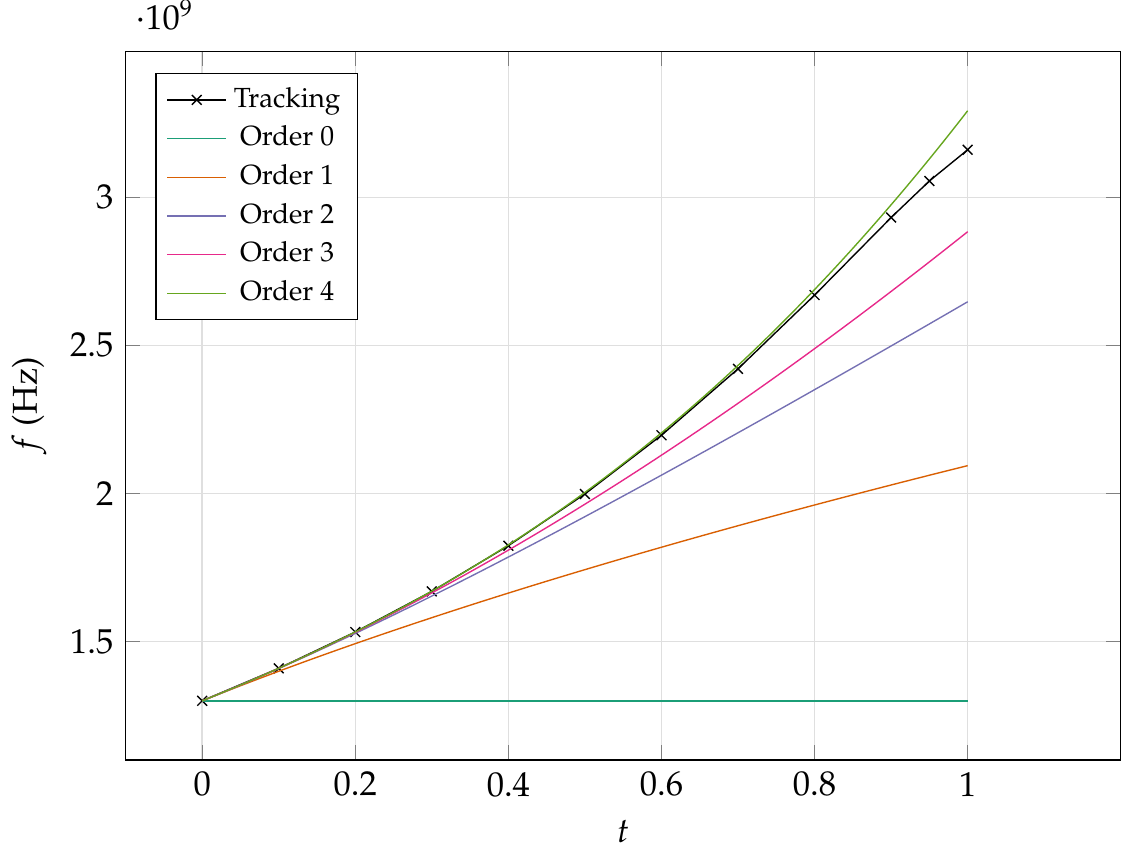}%
   
    \caption{Tracking plot from TESLA cavity ($t=0$) to pillbox cavity ($t=1$). The black line indicates the tracking result achieved with the method described in~\cite{Ziegler_2022ab}. We compare this to the Taylor expansions of orders up to four at $t=0$.}
    \label{fig:trackingPlot}
\end{figure}
The plot {in Fig.~\ref{fig:trackingPlot}} depicts the eigenfrequency along the tracking from the TESLA cavity (at $t=0$) to the pillbox cavity (at $t=1$).
We compare the results of the multi-step tracking method, for which the results of each solved eigenvalue problem are marked in the plot with black crosses, with the estimations we perform by Taylor expansions of orders zero to four at the initial point $t=0$.
It can be clearly observed that the higher order Taylor expansions significantly improve the estimation of the eigenvalue along the tracking curve.
\begin{figure}
\begin{subfigure}[t]{0.45\textwidth}
\centering
\includegraphics[scale=1]{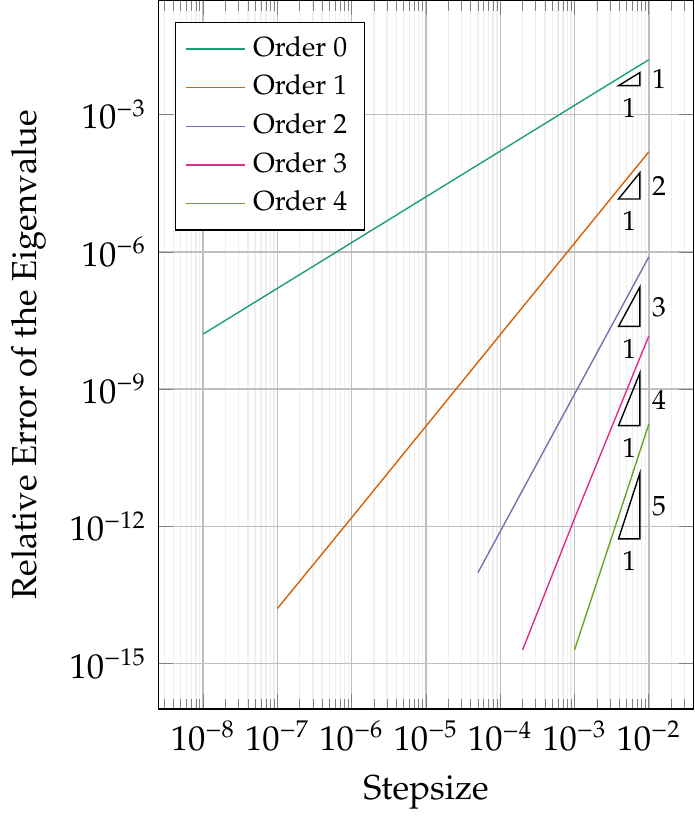}
\caption{Relative error of the Taylor expansions on the eigenvalue of the accelerating eigenmode.}
\label{fig:trackingErrorEigenvalues}
\end{subfigure}\hfill
\begin{subfigure}[t]{0.45\textwidth}
\centering
\includegraphics[scale=1]{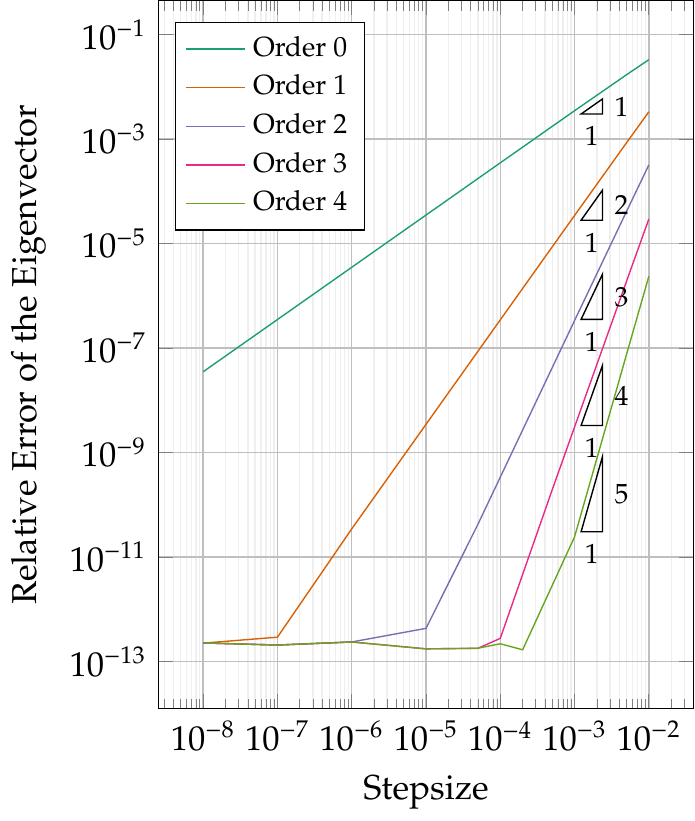}
\caption{Relative error of the Taylor expansions on the eigenvector of the accelerating eigenmode.}
\label{fig:trackingErrorEigenvectors}
\end{subfigure}
\caption{Relative errors of the Taylor expansions on the accelerating eigenmode compared to the numerical solutions at each point.}
\label{fig:trackingError}
\end{figure}

In the convergence plots in Fig.~\ref{fig:trackingError} we investigate the relative error of the Taylor expansions on the eigenvalues and on the eigenvectors, respectively, to the numerical solutions for small step sizes {around $t=0$, i.e., for the TESLA cavity}.
To compute the relative error on the eigenvectors, we take the Euclidean norm of the error and normalize by the norm of the numerical eigenvector at the corresponding step.
In both convergence plots, we notice that the relative errors converge with the expected rates. 

\section{Conclusion}\label{sec:conclusion}
{In this paper, we derived higher order sensitivities of the eigenpairs for the Laplace and the Maxwell eigenvalue problem. 
We have shown how the system matrices for a parameter-dependent domain can be computed on a reference domain of choice, as well as their derivatives with respect to the deformation.
Discretizing the problem with IGA makes the formulation of the domain transformation straightforward, also when considering large shape deformations.
We have shown that our formulation also enables the computation of higher derivatives in closed-form.
Enhancing the Taylor series expansions by higher order terms allows for efficient eigenvalue tracking along a shape morphing as well as for efficient uncertainty quantification. 
}

\section*{Acknowledgment}
The authors thank Philipp Jorkowski and Rolf Schuhmann for their support and the many fruitful discussions, as well as for supplying us with their implementation of the eigenvalue derivative computation \eqref{eq:eigprobderiv}. {Furthermore, the authors thank Annalisa Buffa and Rafael Vázquez for the fruitful discussions.}

This work is supported by the Graduate School CE within the Centre for Computational Engineering at Technische Universität Darmstadt, the Federal Ministry of Education and Research (BMBF) and the state of Hesse as part of the NHR Program and the SFB TRR361/F90 CREATOR (grant number 492661287) funded by the German and Austrian Research Foundations DFG and FWF. Moreover, support by the FWF funded project P32911 and DFG funded project SCHO 1562/6-1 is acknowledged.

\newcommand{\matA}{\ensuremath{\mathbf{X}}}
\newcommand{\matB}{\ensuremath{\mathbf{Y}}}

\appendix
\section{Proof of Lemma \ref{lem_derivJacobi}}\label{apdx:proof_dAdt}
\begin{proof}
    A proof of \eqref{lem_derivJacobi} can be found in, e.g. \cite{Sturm_2015ab}. We give the proof here for the sake of completeness.
    
    In order to see identity \eqref{eq_ddt_partialFinv}, we observe that, for two matrices $\matA$, $\matB$ with $\matA$ invertible, $\partial \mathrm{inv}(\matA)(\matB) = -\matA^{-1} \matB \matA^{-1}$ where $\textrm{inv}(\matA) = \matA^{-1}$. Here, $\partial \mathrm{inv}(\matA)(\matB)$ is to be understood as the derivative of the matrix inversion operator evaluated at a matrix $\matA$ in the direction of a matrix $\matB$. Thus we get
    \begin{align*}
        \ddt ( \textrm{inv}(\partial \maptil[t]) ) =& \partial \textrm{inv}( \partial \maptil[t]) \left( \ddt\partial \maptil[t]\right) = -\partial \maptil[t]^{-1} \left( \ddt\partial \maptil[t]\right) \partial \maptil[t]^{-1}.
    \end{align*}
    
     For identity \eqref{eq_ddt_det}, we use the Jacobi formula for differentiating determinants
        \begin{align*}
            \ddt \det(\partial \maptil[t]) = \tr\left(\textrm{adj}(\partial \maptil[t]) \ddt \partial \maptil[t]\right),
        \end{align*}
    where $\textrm{adj}(\matA)$ denotes the adjugate of a matrix $\matA$. The well-known formula for the adjugate of an invertible matrix, $\textrm{adj}(\matA) =  \det(\matA) \matA^{-1}$, now yields
     \begin{align*}
        \ddt \det(\partial \maptil[t]) = \tr \left(\det(\partial \maptil[t]) (\partial \maptil[t])^{-1} \ddt \partial \maptil[t]\right) 
        =\tr\left( \ddt \partial \maptil[t] (\partial \maptil[t])^{-1} \right) \det(\partial \maptil[t])
    \end{align*}
    where we also used that $\tr(\alpha \matA) = \alpha \tr(\matA)$ and $\tr(\matA\matB ) = \tr(\matB\matA)$.
    
    For \eqref{eq_ddt_A}, we get by the chain rule
    \begin{align*}
        \ddt\mathbf A[t] =& \ddt \left( \det(\partial \maptil[t]) \partial \maptil[t]^{-1} \partial \maptil[t]^{-\top} \right) \\
        =&   \ddt \left( \det(\partial \maptil[t]) \right) \partial \maptil[t]^{-1} \partial \maptil[t]^{-\top} +  \det(\partial \maptil[t]) \ddt \left( \partial \maptil[t]^{-1}\right) \partial \maptil[t]^{-\top}  \\
        &+  \det(\partial \maptil[t]) \partial \maptil[t]^{-1}  \ddt \left(\partial \maptil[t]^{-\top} \right)
    \end{align*}
    For the first term, we get by \eqref{eq_ddt_det}
    \begin{align*}
        \ddt \left(\det(\partial \maptil[t]) \right) \partial \maptil[t]^{-1} \partial \maptil[t]^{-\top} = \mbox{tr} \left( \left(\ddt \partial \maptil[t] \right) \partial \maptil[t]^{-1} \right) \mathbf A[t].
    \end{align*}
    For the second term we get by \eqref{eq_ddt_partialFinv}
    \begin{align*}
         \det(\partial \maptil[t]) \ddt \left( \partial \maptil[t]^{-1}\right) \partial \maptil[t]^{-\top} =& - \det(\partial \maptil[t])  (\partial \maptil[t])^{-1} \left(\ddt \partial \maptil[t] \right) (\partial \maptil[t])^{-1} \partial \maptil[t]^{-\top} \\
         =& - (\partial \maptil[t])^{-1} \left(\ddt \partial \maptil[t] \right) \mathbf A[t].
    \end{align*}
    Similarly, we get for the third term
    \begin{align*}
        \det(\partial \maptil[t]) \partial \maptil[t]^{-1}  \ddt \left(\partial \maptil[t]^{-\top} \right) =& -\det(\partial \maptil[t]) \partial \maptil[t]^{-1}   (\partial \maptil[t])^{-\top} \left(\ddt \partial \maptil[t] \right)^\top (\partial \maptil[t])^{-\top} \\
        =& -\mathbf A[t] \left(\ddt \partial \maptil[t]\right)^\top (\partial \maptil[t])^{-\top} \\
        =& - \left( (\partial \maptil[t])^{-1} \ddt \partial \maptil[t] \mathbf A[t] \right)^\top,
    \end{align*}
    where we used the fact that $\mathbf A[t] = \mathbf A[t] ^\top$ in the last step.
\end{proof}

\section{Proof of Lemma \ref{lem:d2dt2C}}\label{apdx:proof_d2Cdt2}
\begin{proof}
{The first identity follows from \eqref{eq_dDet} and the chain rule.}
In order to compute the second derivative of $\mathbf{C}(t)$, we use that $\ddt ( \tr (\mathbf{X}))) = \tr(\ddt \mathbf{X})$. 
Hence, it holds that
\begin{align*}
\ddt \tr(\dV \dGinv ) = \tr (\dV \ddt \dGinv) = -\tr\left[(\dV \dGinv) (\dV \dGinv)\right],
\end{align*}
and, by the product rule,
\begin{align*}
\frac{\mathrm{d}^2}{\mathrm{d}t^2}\left[ \frac{1}{\det(\dG)} \right] =& \ddt \left[ - \frac{1}{\det(\dG)} \tr(\dV \dG^{-1} ) \right] \\
=& - \left[- \frac{1}{\det(\dG)} \tr(\dV \dG^{-1} ){^2} - {\frac{1}{\det(\dG)}} \tr\left[(\dV \dGinv) (\dV \dGinv)\right]  \right]\\
=& \frac{1}{\det(\dG)} \left[ \tr(\dV \dG^{-1} ){^2} + \tr\left[(\dV \dGinv) (\dV \dGinv)\right]\right].
\end{align*}
Then it holds
\begin{align*}
\frac{\mathrm{d}^2}{\mathrm{d}t^2} \mathbf{C}(t) =& - \ddt(\tr(\dV \dGinv )) \mathbf{C}(t) - \tr(\dV \dGinv ) \ddt \mathbf{C}(t) \\
&+ \ddt\left(\frac{1}{\det(\dG)}\right) \dV^\top \dG + \frac{1}{\det(\dG)} \dV^\top \ddt \dG \\
&+ \ddt\left(\frac{1}{\det(\dG)}\right) \dG^\top \dV + \frac{1}{\det(\dG)}\ddt \dG^\top  \dV \\
=& - \ddt(\tr(\dV \dGinv )) \mathbf{C}(t) - \tr(\dV \dGinv ) \ddt \mathbf{C}(t) \\
&+ \ddt\left(\frac{1}{\det(\dG)}\right) \left[ \dV^\top \dG + \dG^\top  \dV \right] \\
&+ \frac{1}{\det(\dG)} \left[\dV^\top \ddt \dG + \ddt \dG^\top  \dV \right] \\
=& \tr\left[(\dV \dGinv) (\dV \dGinv)\right] \mathbf{C}(t) -  \tr(\dV \dGinv ) \ddt \mathbf{C}(t)\\
&- \frac{1}{\det(\dG)} \mbox{tr}(\dV \dG^{-1} ) \cdot \left[ \dV^\top \dG + \dG^\top  \dV \right]\\
&+ \frac{1}{\det(\dG)} \cdot \left[ \dV ^\top \dV + \dV ^\top \dV \right].
\end{align*}
\end{proof}

\section{Proof of Lemma \ref{lem:d3dt3C}}\label{apdx:proof_dnCdtn}
\begin{proof}
For the computation of higher order derivatives, we apply the General Leibniz rule:
\begin{equation}
(f_1 f_2 \dots f_m)^{(n)} = \sum_{k_1+k_2+ \dots + k_m = n} {\binom{n}{k_1, k_2, \dots , k_m}}  \prod_{1 \leq t \leq m} f_t^{(k_t)}
\end{equation}
We define $\mathbf{C}(t):= f_1 \cdot f_2 \cdot f_3$, hence
\begin{equation*}
f_1 = \funi, \qquad f_2 = \fii,  \qquad f_3 = \fiii
\end{equation*}
and we compute the derivatives individually
\begin{align*}
f_1 &= \funi , \qquad
\ddt f_1 = \fid\\
\frac{\mathrm{d}^2}{\mathrm{d}t^2} f_1 &= \fidd
\end{align*}
\begin{equation*}
f_2 = \fii, \qquad \ddt f_2 = \fiid, \qquad \frac{\mathrm{d}^2}{\mathrm{d}t^2} f_2 = 0,
\end{equation*}
and
\begin{equation*}
f_3 = \fiii, \qquad \ddt f_3 = \fiiid, \qquad \frac{\mathrm{d}^2}{\mathrm{d}t^2} f_3 = 0.
\end{equation*}
To compute the derivatives with the general Leibniz rule, we must first find all $m$-tuples, that sum up to $n$. 
Since the second derivatives of the functions $f_2$ and $f_3$ are zero, the tuples of the corresponding terms can be neglected, since the terms vanish. 
This can be seen in the following table.
\begin{table}[htbp]
\centering
\begin{tabular}{l|l|l}
& tuples & vanishing tuples, since $\f22 = \f32 = 0$ \\ \hline
n = 1 & (1,0,0), (0,1,0), (0,0,1) & \\
n = 2 & (0,1,1), (1,0,1), (1,1,0), (2,0,0);  & (0,2,0), (0,0,2)\\
n = 3 & (3,0,0), (1,1,1), (2,1,0), (2,0,1)  & (0,2,1), (1,0,2), (0,1,2), (0,3,0), (0,0,3), (1,2,0)\\
\end{tabular}
\end{table}
If we simplify the general Leibniz rule for the case $m=3$, then it writes
\begin{equation}
(f_1 f_2 f_3)^{(n)} = \sum_{k_1+k_2+k_3 = n} {\binom{n}{k_1, k_2, k_3}}  \prod_{1 \leq t \leq 3} f_t^{(k_t)}.
\end{equation}
We first set $n=1$ and $n=2$ compute the first and second derivative, respectively, as
\begin{align*}
\ddt \mathbf{C}(t) &= \frac{1!}{1!0!0!} \f11 \f20 \f30 + \frac{1!}{0!1!0!} \f10 \f21 \f30 + \frac{1!}{0!0!1!} \f10 \f20 \f31 \\
&=  \f11 \f20 \f30 + \f10 \f21 \f30 + \f10 \f20 \f31 \\
&= \fid \fii \fiii + \funi \fiid \fiii + \funi \fii \fiiid.
\end{align*}
and
\begin{align*}
\frac{\mathrm{d}^2}{\mathrm{d}t^2} \mathbf{C}(t) =& \frac{2!}{0!1!1!} \f10 \f21 \f31 + \frac{2!}{1!0!1!} \f11 \f20 \f31 + \frac{2!}{1!1!0!} \f11 \f21 \f30 + \frac{2!}{2!0!0!} \f12 \f20 \f30\\
=& 2 \f10 \f21 \f31 + 2 \f11 \f20 \f31 + 2 \f11 \f21 \f30 + \f12 \f20 \f30\\
=& 2 \cdot \funi \cdot \fiid \fiiid  {- 2 \cdot \frac{1}{\det(\dG)} \tr(\dV \dG^{-1} )} \fii \fiiid \\
&- 2 \cdot {\frac{1}{\det(\dG)} \tr(\dV \dG^{-1} )} \fiid \fiii \\
&+ \fidd \fii \fiii.
\end{align*}
This confirms the previous results. 

Analogously, for the third derivative it holds
\begin{align*}
\frac{\mathrm{d}^3}{\mathrm{d}t^3} \mathbf{C}(t) =& \frac{3!}{3!0!0!} \f13 \f20 \f30 + \frac{3!}{1!1!1!} \f11 \f21 \f31 + \frac{3!}{2!1!0!} \f12 \f21 \f30 + \frac{3!}{2!0!1!} \f12 \f20 \f31 \\
=& \f13 \f20 \f30 + 6 \f11 \f21 \f31 + 3 \f12 \f21 \f30 + 3 \f12 \f20 \f31 \\
=& \Bigg(-\tr(\dV\dGinv)\cdot \fidd   \\
& -2 \cdot \frac{1}{\det(\dG)} \cdot \tr(\dV \dGinv)\cdot \tr(\dV \dGinv \dV \dGinv) \\
& - \frac{1}{\det(\dG)} \cdot \tr\Big[(\dV \dGinv \dV \dGinv)(\dV \dGinv)+\dV \dGinv \dV \dGinv \dV \dGinv\Big] \Bigg) \fii  \fiii \\
&- 6 \cdot {\frac{1}{\det(\dG)} \tr(\dV \dG^{-1} )} \fiid \fiiid \\
&+ 3 \cdot \fidd \fiid \fiii\\
&+ 3 \cdot \fidd \fii \fiiid.
\end{align*}
\end{proof}

\end{document}